\documentclass{article}

\usepackage{microtype}

\usepackage{amssymb}
\usepackage{amsmath}
\usepackage{amsthm}
\usepackage{hyperref}
\usepackage{alltt}
\usepackage{color}
\usepackage{graphicx}
\usepackage{authblk}

\title{Between Treewidth and Clique-width}

\author{Sigve Hortemo S\ae{}ther}
\author{Jan Arne Telle}
\affil{Department of Informatics, University of Bergen, Norway\\
  \texttt{\{telle,sigve.sether\}@ii.uib.no}}

\theoremstyle{plain}

\newtheorem{theorem}{Theorem}
\newtheorem{lemma}[theorem]{Lemma}
\newtheorem{observation}[theorem]{Observation}
\newtheorem{corollary}[theorem]{Corollary}
\newtheorem{claim}[theorem]{Claim}
\newtheorem{proposition}[theorem]{Proposition}

\newenvironment{qedproof}{\begin{proof}}{
\end{proof}}

\newcommand{\better}{\preceq}%
\newcommand{\tabme}{\=xx\=xx\=xx\=xx\=xx\=xx\=xx\=xx\=\kill}
\newcommand{\edsc}{\mathcal{F}} %

\newcommand{\cutsize}{\delta}
\newcommand{\cert}{\operatorname{cert}}%
\newcommand{\conc}{\oplus}%

\DeclareMathOperator*{\argmax}{argmax}

\newcommand{\bigoh}{\mathcal{O}}
\newcommand{\mm}{\operatorname{mm}}
\newcommand{\mmw}{\operatorname{mmw}}
\newcommand{\mso}{\operatorname{MSO}}
\newcommand{\tw}{\operatorname{tw}}
\newcommand{\act}{\operatorname{act}}
\newcommand{\smw}{\operatorname{smw}}
\newcommand{\sm}{\operatorname{sm}}
\newcommand{\SMW}{sm-width}
\newcommand{\SM}{sm-value}
\newcommand{\card}[1]{\left\lvert #1 \right\rvert}
\newcommand{\weight}[1]{\card{\act(#1:P)}}
\newcommand{\weightin}[2]{\card{\act(#1:P) \cap #2}}
\newcommand{\tot}{\operatorname{tot}}
\newcommand{\compl}[1]{\overline{#1}}
\newcommand{\cut}[1]{\left(#1,\compl{#1}\right)}
\newcommand{\scut}[1]{({#1}, \compl {#1})}
\newcommand{\set}[1]{\left\{#1\right\}}
\newcommand{\OumTime}{\bigoh^{*}(8^k)}

\begin{document}
\maketitle

\begin{abstract}
  Many hard graph problems can be solved efficiently when restricted
  to graphs of bounded treewidth, and more generally to graphs of
  bounded clique-width. But there is a price to be paid for this
  generality, exemplified by the four problems {\sc MaxCut}, {\sc
    Graph Coloring}, {\sc Hamiltonian Cycle} and {\sc Edge Dominating
    Set} that are all FPT parameterized by treewidth but none of which
  can be FPT parameterized by clique-width unless \mbox{FPT = W[1]}, as shown
  by Fomin et al \cite{FGLS,fomin2010intractability}%
  \footnote{In \cite{fomin2010intractability} the problems {\sc Graph
      Coloring}, {\sc Hamiltonian Cycle} and {\sc Edge Dominating Set}
    are shown to be FPT only if \mbox{FPT = W[1]}, while in \cite{FGLS}
    the authors focus on showing that neither of {\sc MaxCut}, {\sc
      Hamiltonian Cycle} and {\sc Edge Dominating Set} can have a
    $f(k)n^{o(k)}$ algorithm parameterized by clique-width unless the
    Exponential Time Hypothethis fails. However, in \cite{FGLS} they
    also show that {\sc MaxCut} is FPT only if \mbox{FPT = W[1]}.}. %
  We therefore seek a structural graph parameter that shares some of
  the generality of clique-width without paying this price.

  Based on splits, branch decompositions and the work of Vatshelle
  \cite{Vatshelle} on Maximum Matching-width, we consider the graph
  parameter sm-width which lies between treewidth and
  clique-width. Some graph classes of unbounded treewidth, like
  distance-hereditary graphs, have bounded sm-width. We show that {\sc
    MaxCut}, {\sc Graph Coloring}, {\sc Hamiltonian Cycle} and {\sc
    Edge Dominating Set} are all FPT parameterized by sm-width.
\end{abstract}

\section{Introduction}

Many hard problems can be solved efficiently when restricted to graphs
of bounded treewidth or even graphs of bounded clique-width.  A
celebrated algorithmic metatheorem of Courcelle \cite{courcelle}
states that any problem expressible in monadic second-order logic
($\mso_2$) is fixed parameter tractable (FPT) when parameterized by
the treewidth of the input graph.  This includes many problems like
{\sc Dominating Set}, {\sc Graph Coloring}, and {\sc Hamiltonian
  Cycle}.  Likewise, Courcelle et al \cite{CMR} show that the subset
of $\mso_2$ problems expressible in $\mso_1$-logic, which does not
allow quantification over edge sets, is FPT parameterized by
clique-width. Originally this required a clique-width expression as
part of the input, but this restriction was removed when Oum and
Seymour \cite{OS} gave an algorithm that, in time FPT parameterized by
the clique-width $k$ of the input graph, finds a
$2^{O(k)}$-approximation of an optimal clique-width expression.

Clique-width is stronger than treewidth, in the sense that bounded
treewidth implies bounded clique-width \cite{CR} but not vice-versa,
as exemplified by the cliques. Can we hope to find a graph width
parameter lying between treewidth and clique-width for which all
$\mso_2$ problems are FPT? Alas no, under the minimal requirement that
cliques should have bounded width, Courcelle et al \cite{CMR} showed
that this would imply P=NP for unary languages.  There are some basic
problems belonging to $\mso_2$ but not $\mso_1$, like {\sc MaxCut},
{\sc Graph Coloring}, {\sc Hamiltonian Cycle} and {\sc Edge Dominating
  Set}. Fomin et al \cite{FGLS,fomin2010intractability} showed that
none of these four problems can be FPT parameterized by clique-width,
unless FPT~=~W[1].  Can we find a graph width parameter lying between
treewidth and clique-width for which at least these four problems are
FPT? Note that one can define trivial parameters having these
properties (e.g. value equal to clique-width if this is at most 3, and
otherwise equal to treewidth) but can we find one yielding new FPT
algorithms for certain natural graph classes?  This is the question
motivating the present paper, and the answer is yes.  We give a
parameter which is low when the graph has low treewidth in local
parts, and where each of these parts are connected together in a dense
manner.

Before explaining our results, let us mention some related work.  A
class of graphs can have bounded treewidth only if it is sparse.
Indeed, the introduction of clique-width was motivated by the desire
to extend algorithmic results for bounded treewidth also to some dense
graph classes. Let us say that a parameter x is weaker than parameter
y, and y stronger than x, if for any graph class, a bound on x implies
a bound on y. Alternatively, x and y are of the same strength, or
incomparable.  Thus, clique-width is stronger than treewidth.  As we
discussed above there are limitations inherent in clique-width and
there have been several suggestions for width parameters weaker than
clique-width but still bounded on some dense graph classes.  In
particular, let us mention four parameters: neighborhood diversity
introduced by Lampis in 2010 \cite{LampisESA2010}, twin-cover
introduced by Ganian in 2011 \cite{GanianIPEC2011}, shrub-depth
introduced by Ganian et al in 2012 \cite{GanianetalMFCS2012}, and
modular-width proposed by Gajarsk{\'y} et al in 2013 \cite{IPEC2013}.
All these parameters are bounded on some dense classes of graphs, all
of them are weaker than clique-width, but none of them are stronger
than treewidth. Modular-width is stronger than both neighborhood
diversity and twin-cover, but incomparable to shrub-depth
\cite{IPEC2013}. {\sc Graph Coloring} and {\sc Hamiltonian Cycle} are
W-hard parameterized by shrub-depth but FPT parameterized by
modular-width, as recently shown by Gajarsk{\'y} et al \cite{IPEC2013}
which also leaves as an open problem the complexity of {\sc MaxCut}
and {\sc Edge Dominating Set} parameterized by modular-width.

In our quest for a parameter stronger than treewidth and weaker than
clique-width, for which the four basic problems {\sc MaxCut}, {\sc
  Graph Coloring}, {\sc Hamiltonian Cycle} and {\sc Edge Dominating
  Set} become FPT, we are faced with two tasks when given a graph $G$
with parameter-value $k$: we need an FPT algorithm returning a
decomposition of width $f(k)$, and we need a dynamic programming
algorithm solving each of the four basic problems in FPT time when
parameterized by the width of this decomposition. The requirement that
the parameter be stronger than treewidth is a guarantee that it shares
this property with clique-width and will capture large tree-like
classes of graphs, also when some building blocks are dense. Arguably
the most natural way to hierarchically decompose a graph are the
so-called branch decompositions, originating in work of Robertson and
Seymour \cite{RSbrancwidth} and used in the definition of both
rank-width \cite{OS} and boolean-width \cite{BTV}, two parameters of
the same strength as clique-width. Branch decompositions over the
vertex set of a graph can be viewed as a recursive partition of the
vertices into two parts, giving a rooted binary tree where each edge
of the tree defines the cut given by the vertices in the subtree below
the edge.  Using any symmetric cut function defined on subsets of
vertices we can define a graph width parameter as the minimum, over
all branch decompositions, of the maximum cut-value over all edges of
the branch decomposition tree.  Recently, Vatshelle \cite{Vatshelle}
gave a cut-function based on the size of a maximum matching, whose
associated graph width parameter, called MM-width, has the same
strength as treewidth.

In Section 2, based on the work of Vatshelle, we define the parameter
split-matching-width, denoted sm-width, by a cut function based on
maximum matching unless the cut is a split, i.e. a complete bipartite
graph plus some isolated vertices.  The sm-width parameter is stronger
than treewidth and weaker than clique-width.  It is also stronger than
twin-cover but incomparable with neighborhood diversity, shrub-depth
and modular-width.  We finish Section 2 by showing that maximum
matching is a submodular cut function. In Section 3 this is used
together with an algorithm for split decompositions by Cunningham
\cite{Cunningham} and an algorithm for branch decompositions based on
submodular cut functions by Oum and Seymour \cite{OS} to design an
algorithm that given a graph $G$ with sm-width $k$ computes a branch
decomposition of sm-width $O(k^2)$, in time $O^*(8^k)$. To our
knowledge the use of split decompositions to compute a width parameter
is novel.

In Section 4, using a slightly non-standard framework for dynamic
programming, we are then able to solve the four basic problems {\sc
  MaxCut}, {\sc Graph Coloring}, {\sc Hamiltonian Cycle} and {\sc Edge
  Dominating Set}, by runtimes $\bigoh^*(8^k), \bigoh^*(k^{5k}),
\bigoh^*(2^{24k^2})$, and $\bigoh^*(3^{5k})$ respectively, when given
a branch decomposition of sm-width $k$.  In Section 5 we show that
some well-known graph classes of bounded clique-width also have
bounded sm-width, e.g. distance-hereditary graphs have clique-width at
most three and sm-width one.  We also show that a graph whose
twin-cover value is $k$ will have sm-width at most $k$, and discuss
classes of graphs where our results imply new FPT algorithms.  In
Section~\ref{sec:conclusion} we give a short summary of our results
and end the paper by some concluding remarks.

\section{Preliminaries}

We deal with finite, simple, undirected graphs $G=(V,E)$ and denote
also the vertex set by $V(G)$ and the edge set by $E(G)$.  For the
subgraph of $G$ induced by $S \subseteq V(G)$ we write $G[S]$, and for
disjoint sets $A, B \subseteq V(G)$ we denote the induced bipartite
subgraph having vertex set $A \cup B$ and edge set $\{uv: u \in A, v
\in B\}$ as $G[A,B]$.
For a set $E'$ of edges, we denote its endpoints by $V(E')$.  For two
graphs $G_1$ and $G_2$, we denote by $G_1 + G_2$ the graph with vertex
set $V(G_1) \cup V(G_2)$ and edge set $E(G_1) \cup E(G_1)$. If a set
of vertices $V'$ or a set of edges $E'$ is written where a graph is
expected (e.g., $G_1 + E'$ or $G_1 + V'$), we interpret $E'$ as the
graph $(V(E'), E')$ and $V'$ as the graph $(V', \emptyset)$.
For $v \in V(G)$ we write $N(v)$ or $N_G(v)$
for the neighbors of $v$ and for $S \subseteq V(G)$ we denote the
neighborhood of $S$ by $N(S) = \bigcup_{a \in S} N(a) \setminus S$ or
$N_G(S)$; note that $N(S) \cap S = \emptyset$.  A matching is a set of
edges having no endpoints in common.

A \emph{split} of a connected graph $G$ is a partition of $V(G)$ into
two sets $V_1, V_2$ such that $|V_1| \geq 2$, $|V_2| \geq 2$ and every
vertex in $V_1$ with a neighbor in $V_2$ has the same neighborhood in
$V_2$ 
(this also means every vertex in $V_2$ with a neighbour in $V_1$ has
the same neighbourhood in $V_1$). %
A graph $G$ with a split $(V_1, V_2)$ can be \emph{decomposed}
into a graph $G_1$ and a graph $G_2$ so that $G_1$ and $G_2$ is the
induced subgraph of $G$ on $V_1$ and $V_2$, respectively, except that
an extra vertex $v$, called a \emph{marker}, is added, and also some
extra edges are added to $G_1$ and $G_2$, so that $N_{G_1}(v) =
N_G(V_2)$ and $N_{G_2}(v) = N_G(V_1)$. If a graph $G$ can be
decomposed to the two graphs $G_1$ and $G_2$, then $G_1$ and $G_2$
\emph{compose} $G$. We denote this by $G = G_1 * G_2$.  A graph that
cannot be decomposed (i.e., a graph without a split) is called a
\emph{prime}. As all graphs of at most three vertices trivially is a
prime, when a prime graph has more than three vertices, it is called a
\emph{non-trivial} prime graph.  A \emph{split decomposition} of a
graph $G$ is a recursive decomposition of $G$ so that all of the
obtained graphs are prime.  For a split decomposition of $G$ into
$G_1, G_2, \ldots, G_k$, a \emph{split decomposition tree} is a tree
$T$ where each vertex corresponds to a prime graph and we have an edge
between two vertices if and only if the prime graphs they correspond
to share a marker. That is, the edge set of the tree is $E(T) =
\{v_iv_j : v_i,v_j \in V(T) \mbox{ and } V(G_i) \cap V(G_j) \neq
\emptyset\}$. To see that this is in fact a tree, we notice that $T$
is connected and that we have an edge for each marker introduced. As
there are exactly one less marker than there are prime graphs, $T$
must be a tree. See Figure \ref{fig:6} for an example.

Given a split decomposition of graph $G$ with prime graphs $G_1, G_2,
\ldots, G_k$, we define $\tot(v:G_i)$ recursively to be $\{v\}$ if
$v \in V(G)$, and otherwise to be $\bigcup_{u \in V(G_j) \setminus
  \{v\}} \tot(u:G_j)$ for the graph $G_j \not= G_i$ containing the marker $v$
in the split decomposition. 
Another way of saying this latter part by
the use of the split decomposition
tree $T$ is: if $v$ is not in $V(G)$, then $\tot(v: G_i)$ is defined
to be the vertices of $V(G)$ residing in the prime graphs of the
connected component in $T[V(T) - G_i]$ where $v$ is also
located.
  From this last definition, we observe that for a prime
graph $G_i$ in a split decomposition of $G$, the function $\tot$
on the vertices of $G_i$ partitions the vertices of $V(G)$. For a set
$V' \subseteq V(G_i)$, we define $\tot(V':G_i)$ to be the union of
$\tot(v:G_i)$ for all $v \in V'$.  For a set $S \subseteq V(G)$, the
inverse function $\tot^{-1}(S:G_i)$, is defined as the minimal set
of vertices $V' \subseteq V(G_i)$ so that $S \subseteq
\tot(V':G_i)$.  We define the \emph{active set} of a vertex $v \in
G_i$, denoted $\act(v:G_i)$ to be the vertices of $\tot(v:G_i)$
that are contributing to the neighborhood of $v$ in $G_i$. That is,
$\act(v:G_i)$ is defined as $N(V(G) \setminus \tot(v:G_i))$. 
See Figure \ref{fig:6} for an
example of $\tot()$ and $\act()$.
Note that if $G$ has a split decomposition into prime graphs $G_1, \ldots,
G_k$, then for any marker $v$ there are exactly two prime graphs $G_i$ and
  $G_j$ containing $v$, and we have $\tot(v:G_i) \cup \tot(v:G_j)
  = V(G)$.

  \begin{figure}[h]
    \begin{minipage}{0.4\linewidth}   
      \centering
      \includegraphics[scale=0.4]{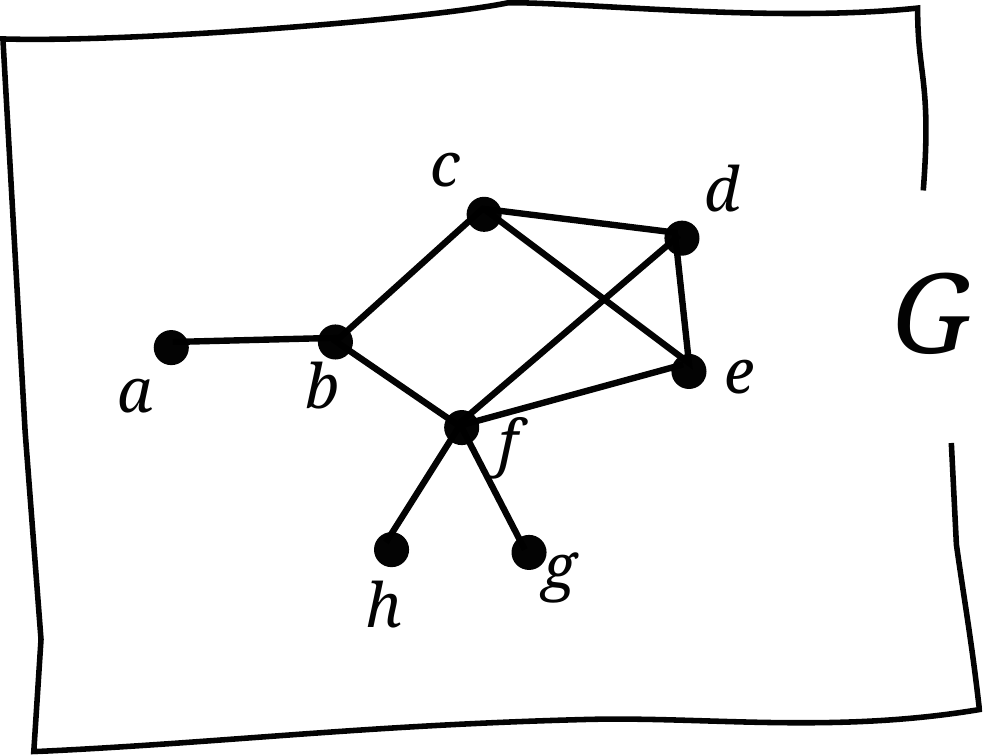}
    \end{minipage}
    \begin{minipage}{0.5\linewidth}
      \includegraphics[scale=0.45]{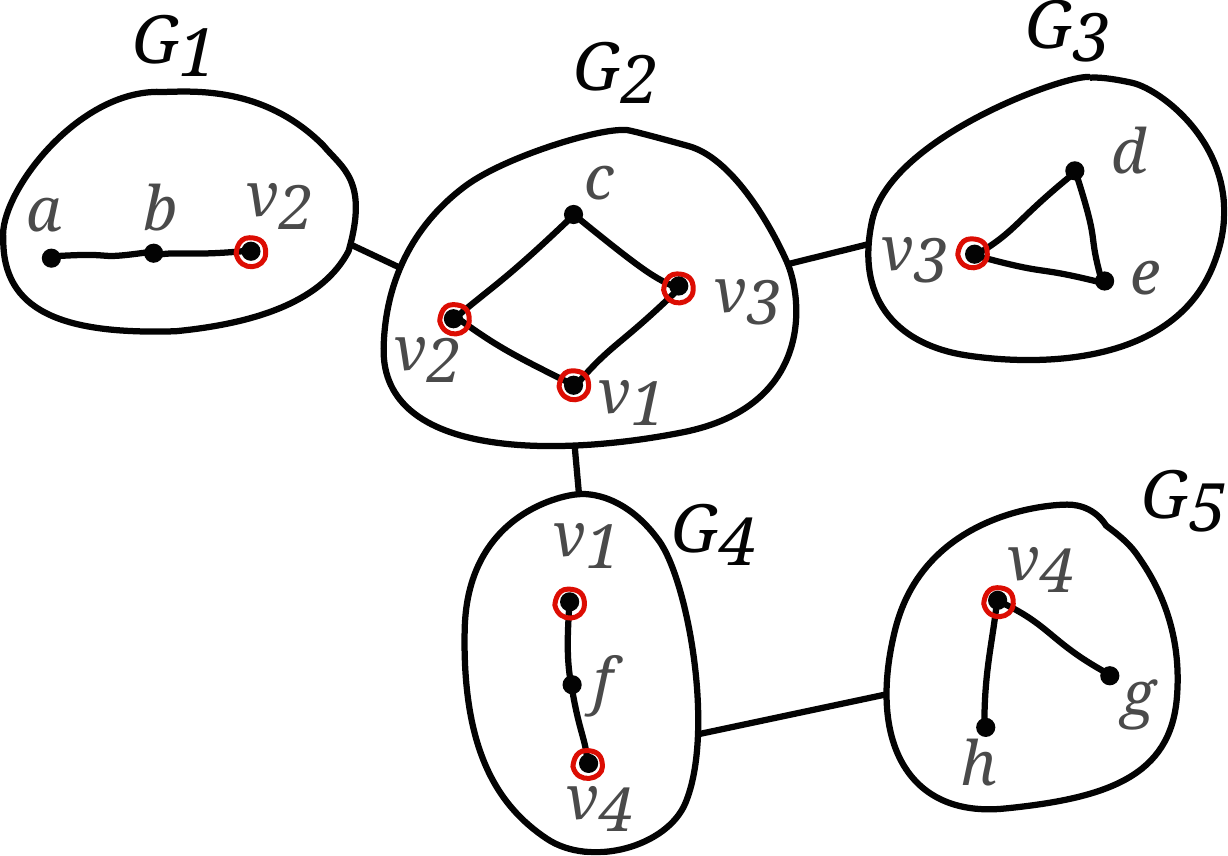}
      \end{minipage}
    \caption{Split decomposition tree of a graph $G$. The markers of each prime graph are circled in red. An
      example of a split decomposition resulting in this tree is: $((G_1
      * G_2) * G_3) * (G_4 * G_5)$. Note that $\tot(\{v_1,v_3\}:G_2) = \{d,e,f,g,h\}$, $\act(v_1:G_4) = \{b,d,e\}$ and
      $\tot^{-1}({a,f,g}:G_2) = \{v_1,v_2\}.$}
    \label{fig:6}
  \end{figure}

  A \emph{branch decomposition} $(T, \delta)$ of a graph $G$ consists
  of a subcubic tree $T$ (a tree of maximum degree $3$) and a
  bijective function $\delta$ from the leaves of $T$ to the vertices
  of $G$.  For a graph $G$ a \emph{cut} $(A, \compl{A})$ for $A
  \subseteq V(G)$ is a bipartition of verices of $G$. For a cut $(A, B)$ of $G$, we say the edges in $G$ with one enpoint in $A$ and the other in $B$ \emph{cross} the cut $(A, B)$. %
  In a branch decomposition $(T=(V_T, E_T), \delta)$ of a graph $G$,
  each edge $e \in E_T$ partitions $V(G)$ into two parts: the vertices
  mapped by $\delta$ from the leaves of one component of $T - e$, and
  the vertices mapped by $\delta$ from the leaves of the other
  component. Thus each edge of $T$ induces a cut in $G$, namely the
  cut corresponding to that edge's bipartition of $V(G)$.  For a graph
  $G$, a \emph{cut function} $f: 2^{V(G)} \rightarrow \mathbb{N}$ is a
  symmetric ($f(A) = f(\compl{A})$) function on subsets of $V(G)$.
For a branch decomposition $(T, \delta)$ of $G$ its
$f$-width, for a cut function $f$, is the maximum of $f(A)$ over all cuts $(A, \compl{A})$ of $G$ 
induced by the edges of $T$. 
For a graph $G$, its $f$-width, for a cut function $f$, is the minimum $f$-width over all branch
decompositions of $G$.

Vatshelle \cite{Vatshelle} defined the Maximum-Matching-width (MM-width) $\mmw(G)$ of a graph $G$ based
on the cut
function $\mm$ defined for any graph $G$ and $A \subseteq V(G)$ by letting
$\mm(A)$ be the cardinality of a maximum matching of the bipartite graph $G[A,\compl A]$.
In his work, Vatshelle shows that there is a linear dependency between
the treewidth of a graph and the Maximum-Matching-width of the graph.
\begin{theorem} [\cite{Vatshelle}] \label{vatshelle}
  Let $G$ be a graph, then $\frac{1}{3}(\tw(G) + 1) \leq \mmw(G) \leq \tw(G) + 1$
\end{theorem}

In this paper we define the split-matching-width $\smw(G)$ of a graph $G$ based
on the cut function $\sm$
defined for any graph $G$ and $A \subseteq V(G)$ by:
\[\sm(A) = 
\begin{cases}
  1  \mbox{ if $\cut{A}$ is a split of $G$} \\
  \mm(A)=\max\{\card{M} : \mbox{$M$ is a matching of $G[A, \compl A]$}\}  \mbox{ otherwise}
\end{cases}\]%
A cut function $f: 2^{V(G)} \rightarrow
\mathbb{N}$ is said to be submodular if
  for any $A,B \subseteq V(G)$ we have $f(A) + f(B) \geq f(A \cup B) +
  f(A \cap B)$.
The following very general result of Oum and Seymour is central to the field of branch decompositions.

\begin{theorem} [\cite{OS}] \label{lemma:oumSeymour}
  For  symmetric  submodular cut-function $f$  and graph $G$  of
  optimal $f$-width $k$, a branch decomposition of $f$-width at most $3k+1$ can be
  found in $\mathcal{O}^*(2^{3k+1})$ time.
\end{theorem}

There is no abundance of submodular cut functions, but this result
 will be useful to us.

\begin{theorem}\label{theorem:mm_submod}
  The cut function $\mm$ is submodular.
\end{theorem}

\begin{qedproof}
  Let $G$ be a graph and $S \subseteq V(G)$. We will say that a
  matching $M \subseteq E(G)$ is a matching of $S$ if each edge of $M$
  has exactly one endpoint in $S$, i.e. $M$ is a matching of the
  bipartite graph $G[S, \compl{S}]$.  To prove that $\mm$ is
  submodular, we will show that for any $A,B \subseteq V(G)$ and any
  matching $M_{A \cup B}$ of $A \cup B$ and $M_{A \cap B}$ of $A \cap
  B$, there must exist two matchings $M_A$ of $A$ and $M_B$ of $B$ so
  that the multiset of edges $M_A \uplus M_B$ is equal to the multiset
  $M_{A \cup B} \uplus M_{A \cap B}$.  First notice that each edge of
  $M_{A \cup B}$ and $M_{A \cap B}$ is a matching of either $A$ or $B$
  (or both).  As the vertices in a matching have degree one, the
  multiset $M_{A \cup B} \uplus M_{A \cap B}$ of edges can be regarded
  as a set of vertex disjoint paths and cycles (note though, we might
  have cycles of size two, as the same edge might be in both of the
  matchings).  We will show that for every such path or cycle $P$
  there exist matchings $M_A^P$ for $A$ and $M_B^P$ for $B$ so that
  $E(P) = E(M_A^P) \cup E(M_B^P)$.  Note that this suffices to prove
  the statement, as there will then also exist matchings $M_A$ of $A$
  and $M_B$ of $B$ so that $E(M_A) \uplus E(M_B) = E(M_{A \cup B})
  \uplus E(M_{A \cap B})$, by taking $M_A$ and $M_B$ as the disjoint
  union of each of the smaller matchings, for $A$ and $B$
  respectively, that exist for each path or cycle $P$ in $M_{A \cup B}
  \uplus M_{A \cap B}$. Since these paths and cycles are
  vertex-disjoint $M_A$ and $M_B$ will be matchings.
  
  Thus, let $P$ be a path or a cycle from $M_{A \cup B} \uplus M_{A \cap B}$. If
  $P$ only contains vertices of $A \cap B$ and $\compl{A \cup B}$,
  each edge of $P$ is a matching of both $A$ and $B$, so we have the
  matchings by setting $M_A = P \cap M_{A \cap B}$ and $M_B = P \cap
  M_{A \cup B}$. %
  Since the edges of
  $E(P)$ alternate between $M_{A \cup B}$ and $M_{A \cap B}$, and
  since all edges from $M_{A \cup B}$ has an endpoint in $\compl{A
    \cup B}$ and all edges from $M_{A \cap B}$ has an endpoint in $A
  \cap B$, there can be at most one vertex $v$ in $P$ belonging to
  $(B \setminus A) \cup (A \setminus B)$
  (it may help to look at 
  Figure \ref{fig:3} where it is clear that no path
  alternating between blue and red edges can touch $(B \setminus A) \cup (A \setminus B)$ twice).
  If there exists such a vertex $v$, assume without loss of generality that $v \in B
  \setminus A$. As each edge in $M_{A \cap B} \cap P$ has exactly one endpoint in
  $A \cap B$, and $P$ contains vertices only of $A \cap B, B \setminus
  A$ and $\compl{A \cup B}$, all the edges of $M_{A \cap B} \cap P$
  has one endpoint in $A$ and one endpoint in $(B \setminus A) \cup
  \compl{A \cup B} = \compl{A}$. So, $M_{A \cap B} \cap P$ is a
  matching of $A$.
  For $M_{A \cup B}$, by the same arguments, each edge in $M_{A \cup
    B} \cap P$ must have one endpoint in $\compl{A \cup B}$ and one
  endpoint in $(B \setminus A) \cup (A \cap B) = B$, making $M_{A \cup
    B} \cap P$ a matching of $B$.
\end{qedproof}

\begin{figure}[h]  
  \centering
  \includegraphics[scale=0.8]
  {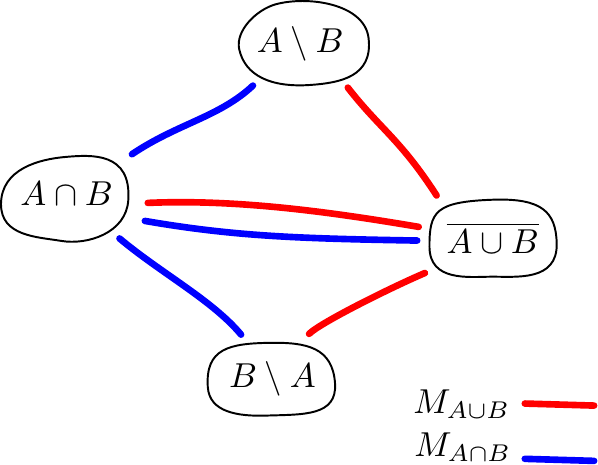}
  \caption{The edges of $M_{A \cap B}$ and $M_{A \cup B}$}
  \label{fig:3}
\end{figure}

\section{Computing an approximate sm-width decomposition}

In this section we design an algorithm that given a graph $G$ finds 
a branch
decomposition of $G$ having sm-width $O(\smw(G)^2)$, in time FPT parameterized by $\smw(G)$.
The algorithm has four main steps:
\begin{enumerate}
 \item Find a split decomposition of $G$ into prime graphs $G_1, G_2,...,G_q$.
 \item For each $G_i$ find a branch decomposition $(T_i, \delta_i)$ of sm-width $O(\smw(G_i))$.
 \item For each $G_i$ restructure $(T_i, \delta_i)$ into $(T_i', \delta_i')$ having the property that any cut of $G_i$,
 induced by an edge of $(T_i', \delta_i')$ and having split-matching value $k$, is lifted, by the split decomposition of
 $G$, to a cut of $G$ having split-matching value $O(k^2)$.
 \item Combine all the decompositions $(T_i', \delta_i')$  into a branch decomposition of $G$ of sm-width $O(\smw(G)^2)$.
\end{enumerate}

For step 1 there exists a well-known polynomial-time algorithm by Cunningham
\cite{Cunningham} and even linear-time ones, see e.g. \cite{Charbit} and see also \cite{Rao}
for the use of split decompositions in general.
For step 2 we are dealing with a prime graph $G_i$, which by definition has
no non-trivial splits and hence $\sm(V_i) = \mm(V_i)$ for all $V_i \subseteq
V(G_i)$ meaning that $\mmw(G_i)=\smw(G_i)$. Furthermore, by Theorem~\ref{theorem:mm_submod} 
the cut function defining $\mmw$ is submodular so we can apply the algorithm of
Oum and Seymour from Theorem~\ref{lemma:oumSeymour} to accomplish the task of step 2. 
Step 3 will require more work. Let us first give a sketch of step 4.
Suppose for each prime graph $G_i$ of a split decomposition of $G$ we
have calculated a branch decomposition $(T_i', \delta_i')$ for $G_i$.
If for every cut
$(X, V(G_i) \setminus X)$ of $G_i$ induced by an edge of  $(T_i', \delta_i')$ we have
$\sm(\tot(X:G_i)) \leq t$ for some value $t$, then we can generate a
branch decomposition of $G$ of \SMW\ at most $t$ by for each pair of prime graphs
$G_i, G_j$ sharing a marker, 
identifying the two leaves of respectively $(T_i', \delta_i')$ and $(T_j', \delta_j')$
mapped to this marker (see figure \ref{fig:2}).

\begin{figure}[h]                                                           
  \begin{minipage}{0.4\linewidth}                                              
    \centering                                                                 
    \fbox{\includegraphics[width=0.95\textwidth,] 
    {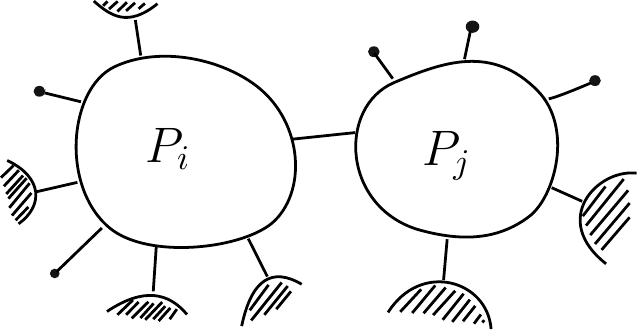}}                                                          
\end{minipage}                                                               
  \begin{minipage}{0.6\linewidth}                                              
    \centering                                                                 
    \fbox{\includegraphics[width=0.90\textwidth] 
    {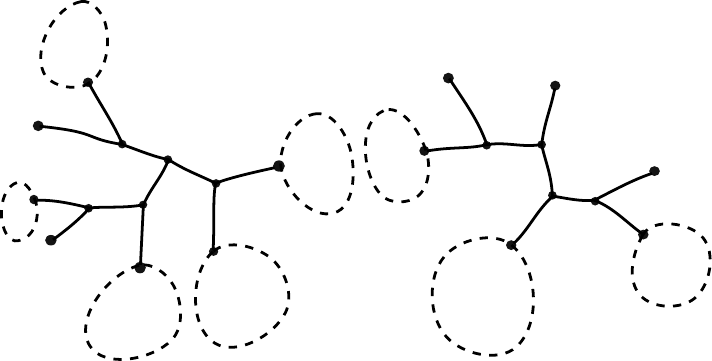}}\\[7pt]                                                   
    $\Downarrow$\\[7pt]                                                        
    \fbox{\includegraphics[width=0.90\textwidth]   
    {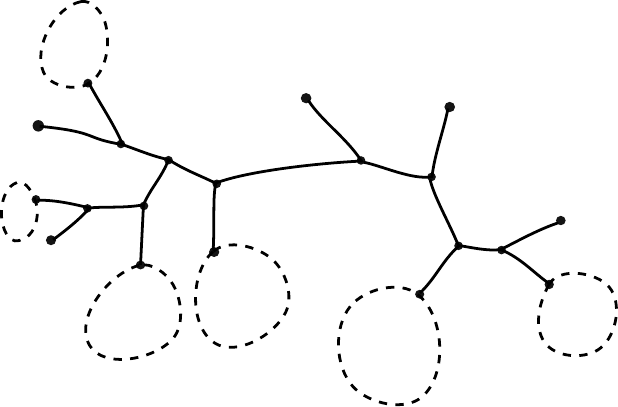}}                                                          
  \end{minipage}                                                               
  \caption{Combining the decomposition of prime graphs sharing a
    marker.  The prime graphs in a split decomposition tree to the
    left and their branch decomposition trees -- before and after
    combining them -- to the right.}
  \label{fig:2}                                                                
\end{figure}  

What remains is step 3, covered by Theorem \ref{thm:PrimBranch}. We
need to relate $\sm(A)$ of a cut $(A, V(G_i) \setminus A)$ in prime
graph $G_i$, induced by an edge of $(T_i', \delta_i')$, to
$\sm(\tot(A:G_i))$ of the associated cut $\cut{\tot(A:G_i)}$ in
$G$. This we do by the series of lemmas from
Lemma~\ref{lemma:1JoinsAreSingleSidedInPrimeGraph} to
Lemma~\ref{lemma:mmRelasjonMellomPogG}. The latter of these lemmas use
the notion of a \emph{heavy} pair of vertices: In a graph $G$ with
$smw(G) < k$ and split decomposition into prime graphs $G_1,
G_2,...,G_q$ we say that adjacent vertices $a,b \in V(G_i)$ are
\emph{heavy} if $\card{\act(a:G_i)} \geq 3k$ and $\card{\act(b:G_i)}
\geq 3k$. The edge connecting a heavy pair is called a \emph{heavy}
edge.

\begin{lemma} \label{lemma:1JoinsAreSingleSidedInPrimeGraph} Let $G$
  be a graph and $P$ a non-trivial prime graph in a split decomposition
  of $G$.  For any split $(X,Y)$ of $G$, there exist a vertex $v \in
  V(P)$ such that either $X \subseteq \tot(v:P)$ or $Y \subseteq
  \tot(v:P)$.
\end{lemma}

\begin{qedproof} Let $X_P = \tot^{-1}(X:P)$ and $Y_P =
  \tot^{-1}(Y:P)$.  Assume for contradiction that both $\card{X_P}
  \geq 2$ and $\card{Y_P} \geq 2$.  Since $X \cup Y = V(G)$, we have $X_P
  \cup Y_P = V(P)$.
  The fact that $X_P \cup Y_P = V(P)$ and $\card{X_P} \geq 2$,
  $\card{Y_P} \geq 2$, $\card{V(P)} \geq 4$ implies that 
  $V(P)$ has a partition into $X_P', Y_P'$ with $X_P' \subseteq X_P$
  and $Y_P' \subseteq Y_P$, and $\card{X_P'} \geq 2$ and $\card{Y_P'}
  \geq 2$.

  As $P[X_P', Y_P']$ is isomorphic to an induced subgraph of $G[X,Y]$,
  and $(X,Y)$ is a split of $G$, $(X_P', Y_P')$ must also be a split
  in $P$.  However, as both $X_P'$ and $Y_P'$ have cardinality at
  least $2$, this contradicts the fact that a prime graph does not
  have any splits.
\end{qedproof}

\begin{lemma} \label{lemma:kkMeans1join} Let $G$ be a graph, $P$ a
  non-trivial prime graph in a split decomposition of $G$ and
  $(T,\delta)$ a branch decomposition of \SMW\ less than $k$.  For $a
  \in V(P)$ and cut $(X,Y)$ in $(T,\delta)$, if $\card{ X \cap
    \act(a:P) } \geq k$ and $\card{ Y \cap \act(N(a):P) } \geq
  k$, then $(X, Y)$ is a split.
\end{lemma}

\begin{qedproof}
  Since $X \cap \act(a:P)$ and $Y \cap \act(N(a):P)$ form a complete bipartite
  graph with at least $k$ vertices on each side, the $mm$-value of the
  cut $(X, Y)$ must be at least $k$.  And since $(X, Y)$ is a cut in a
  branch decomposition of \SMW\ less than $k$, we conclude that $(X,Y)$
  must be a split.
\end{qedproof}

\begin{lemma} \label{lemma:branch3divide} For any two (not necessarily
  disjoint) vertex subsets $A$ and $B$ of $V(G)$, and in any branch
  decomposition $(T, \delta)$ of $G$, there must exist a cut $(X,Y)$
  in $(T, \delta)$ so that $\card{X \cap A} \geq
  \left\lfloor\frac{\card{A}}{3}\right\rfloor$ and $\card{Y \cap B}
  \geq \left\lfloor\frac{\card{B}}{3}\right\rfloor$.
\end{lemma}

\begin{qedproof}
  For a single $S \subseteq V(G)$ it is well known that since $T$ is a
  ternary tree with leaf set $V(G)$ there exists a cut $(X_S,Y_S)$ in
  $(T, \delta)$ associated with an edge $(x_S,y_S) \in E(T)$ so that
  $|X_S \cap S| \geq \lfloor \frac{|S|}{3} \rfloor$ and $|Y_S \cap S|
  \geq \lfloor \frac{|S|}{3} \rfloor$.  Consider the path in $T$
  starting in edge $(x_A,y_A)$ and ending in edge $(x_B,y_B)$.  The cut
  associated with any edge on this path will satisfy the statement in
  the lemma.
\end{qedproof}

\begin{lemma} \label{lemma:9kLemma} Let $G$ be a graph, $P$ a
  non-trivial prime graph in a split decomposition of $G$ and $(T,
  \delta)$ a branch decomposition of \SMW\ less than $k$.  If $P$ has
  vertex $b$ such that $\weight{b} \geq 3k$ and $\weight{N(b)} \geq
  9k$, then there must exist vertex $a \in N(b)$ such that $a$ and $b$
  form a heavy pair.
\end{lemma}

\begin{qedproof}
  Assume for contradiction that this is not the case, and all vertices
  $v \in N(b)$ have $\weight{v} < 3k$.  By
  Lemma~\ref{lemma:branch3divide} applied to $A=\act(b:P)$ and $B=\act(N(b):P)$,
  there must be a cut $(X,Y)$ in $(T,\delta)$ so that $\card{\act(b:P) \cap
    X} \geq k$ and $\card{\act(N(b):P) \cap Y} \geq 3k$.  We first show
  that $\act(N(b):P) \subseteq Y$ and thus $\card{\act(N(b):P) \cap Y} \geq 9k$.

  By Lemma~\ref{lemma:kkMeans1join}, $(X,Y)$ is a split.  Since no $v
  \in N(b)$ has $\weight{v} \geq 3k$ the fact that $\card{Y \cap
    \act(N(b):P)} \geq 3k$ means that $\card{\tot^{-1}(Y:P)} \geq 2$.  By
  Lemma~\ref{lemma:1JoinsAreSingleSidedInPrimeGraph}, this means $X
  \subseteq \tot(b:P)$ and thus $\tot(N(b):P) \subseteq V \setminus X = Y$
  and thus $\act(N(b):P) \subseteq Y$.

  Again, by Lemma~\ref{lemma:branch3divide}, applied to $\act(N(b):P)=A=B$,
  there must exist a cut $(X',Y')$ in $(T, \delta)$ so that $\card{X'
    \cap \act(N(b):P)} \geq 3k$ and $\card{Y' \cap \act(N(b):P)} \geq 3k$.  Since
  we have already shown that $\act(N(b):P) \subseteq Y$ and both $(X, Y)$
  and $(X',Y')$ are cuts in $(T, \delta)$, either $X \subseteq X'$ or
  $X \subseteq Y'$.  Without loss of generality, let us assume $X
  \subseteq X'$.  Since there are at least $3k$ vertices of $\act(N(b):P)$
  in $Y'$ and at least $k$ vertices of $\act(b:P)$ in $X \subset X'$, by
  Lemma~\ref{lemma:kkMeans1join} $(X',Y')$ must be a split.  The
  assumption that all $v \in N(b)$ have $\weight{v} < 3k$ means that
  $\tot^{-1}(X':P)$ and $\tot^{-1}(Y':P)$ contain at least two vertices
  from $\act(N(b):P)$, which together with the fact that $(X', Y')$ is a
  split is a contradiction of
  Lemma~\ref{lemma:1JoinsAreSingleSidedInPrimeGraph}.  Therefore, our
  assumption was wrong, and there must exist vertex $a \in N(b)$ such
  that $\weight{a} \geq 3k$.
\end{qedproof}

\begin{lemma} \label{lemma:abSubtree} Let $G$ be a graph, $P$ a
  non-trivial prime graph in a split decomposition of $G$ and $(T,
  \delta)$ a branch decomposition of \SMW\ less than $k$.  For any
  heavy pair $a,b$ in $P$ we have that there must exist a cut $(X,Y)$ in $(T,
  \delta)$ where $\tot^{-1}(X:P) = \set{a,b}$ and %
  $\card{N(X)} < k$.
\end{lemma}

\begin{qedproof} By the definition of a heavy pair, $|\act(a:P)| \geq
  3k$ and $|\act(b:P)| \geq 3k$. We know, by
  Lemma~\ref{lemma:branch3divide} on $\act(a:P)$ and $\act(b:P)$ that
  there must exist a cut $\cut{X_1}$ in $(T, \delta)$ where
  $\weightin{b}{X_1} \geq k$ and $\weightin{a}{\compl{X_1}} \geq k$.
  By Lemma~\ref{lemma:kkMeans1join}, this cut must be a split, and by
  Lemma~\ref{lemma:1JoinsAreSingleSidedInPrimeGraph} either $X_1
  \subseteq tot(b)$ or $\compl{X_1} \subseteq tot(a)$.  Without loss
  of generality, let $X_1 \subseteq tot(b)$.  This means that
  $\act(a:P) \subseteq \tot(a:P) \subseteq \compl{X_1}$.  By
  Lemma~\ref{lemma:branch3divide} on $\act(a:P)$, we know there is a
  cut $\cut{X_2}$ in $(T, \delta)$ so that $\weightin{a}{X_2} \geq k$
  and $\weightin{a}{\compl{X_2}} \geq k$.  Since $\act(a:P) \subseteq
  \compl{X_1}$ and both $\cut{X_1}$ and $\cut{X_2}$ are cuts in $(T,
  \delta)$, either $X_1 \subseteq X_2$ or $X_1 \subseteq \compl{X_2}$.
  Without loss of generality, let $X_1 \subseteq X_2$.  This means
  $\compl{X_2}$ contains vertices of $\act(a:P)$ and of $\act(b:P)$,
  and thus $\card{\tot^{-1}(X_2:P)} \geq 2$.  However, $\cut{X_2}$ is
  a split and $\compl{X_2} \cap \tot(a:P) \neq \emptyset$, so by
  Lemma~\ref{lemma:1JoinsAreSingleSidedInPrimeGraph}, $\compl{X_2}
  \subseteq \tot(a:P)$.

  Now, let  $A$ and $B$ be  the largest sets  $\compl{X_2} \subseteq A
  \subseteq \tot(a:P)$  and $X_1 \subseteq B  \subseteq \tot(b:P)$, so
  that there exist  cuts $\cut{A}$ and $\cut{B}$ in  $(T, \delta)$ (an
  equivalent  way  of  saying  this  is  that  $\cut{B}$  is  the  cut
  associated with the last edge along the path in $T$ from $\cut{X_1}$
  to $\cut{X_2}$ so that $B \subseteq \tot(b:P)$, and $\cut{A}$ is the
  cut associated  with the last  edge from $\cut{X_2}$  to $\cut{X_1}$
  where $A \subseteq \tot(a:P)$).  By the above, such sets must exist.
  Since $T$ is a cubic tree, the edge in $E(T)$ representing $\cut{B}$
  must be adjacent to two  other edges that represent two cuts ${(R_1,
    B \cup  R_2)}$ and  ${(R_2, B \cup  R_1)}$ where $\compl{B}  = R_1
  \cup  R_2$, as depicted  in Figure  \ref{fig:1}.  As  $B \cap  A$ is
  empty, and  there is  a cut $\cut{A}$,  we have either  $A \subseteq
  R_1$ or $A \subseteq R_2$.   Without loss of generality $A \subseteq
  R_2$.   By  maximality  of  $B  \subseteq \tot(b:P)$  we  have  $R_1
  \setminus \tot(b:P) \neq \emptyset$ and thus $\card{\tot^{-1}(B \cup
    R_1:P)} \geq 2$.  However, $(B \cup  R_1, R_2)$ must be a split by
  Lemma~\ref{lemma:kkMeans1join},  since   there  are  at   least  $k$
  vertices  of $\act(a:P)$  in  $A  \subseteq R_2$  and  at least  $k$
  vertices  of  $\act(b:P)$  in   $(B  \cup  R_1)$.   Furthermore,  by
  Lemma~\ref{lemma:1JoinsAreSingleSidedInPrimeGraph},  $R_2  \subseteq
  \tot(a:P)$.   However,  as $A  \subseteq  R_2$  is  the largest  set
  $\compl{X_2} \subseteq  A \subseteq \tot(a:P)$ so  that $\cut{A}$ is
  in $(T,  \delta)$, we have $A =  R_2$.  This means the  cut $(R_1, B
  \cup R_2)$ is  in fact $(R_1, B \cup A)$  and thus $\tot^{-1}(B \cup
  A:P_G) = \set{a,b}$.

  Furthermore, since $\card{\tot^{-1}(B \cup A:P)} = 2$ and
  $\card{V(P)} \geq 4$, by
  Lemma~\ref{lemma:1JoinsAreSingleSidedInPrimeGraph} $(R_1, B \cup A)$
  cannot be a split.  The MM-value of a cut is the same as the
  smallest vertex cover of the bipartite graph associated with cut.
  As all minimal vertex covers of $G[R_1, A \cup B]$ contain
  either at least $k$ vertices of $\act(b:P) \cap B$, at least $k$
  vertices of $\act(a:P) \cap A$, or all the neighbors of $A$ and $B$
  in $R_1$, we conclude that $N(A \cup B) < k$.
\end{qedproof}

\begin{figure}[h]
  \centering
  \includegraphics[scale=0.8]{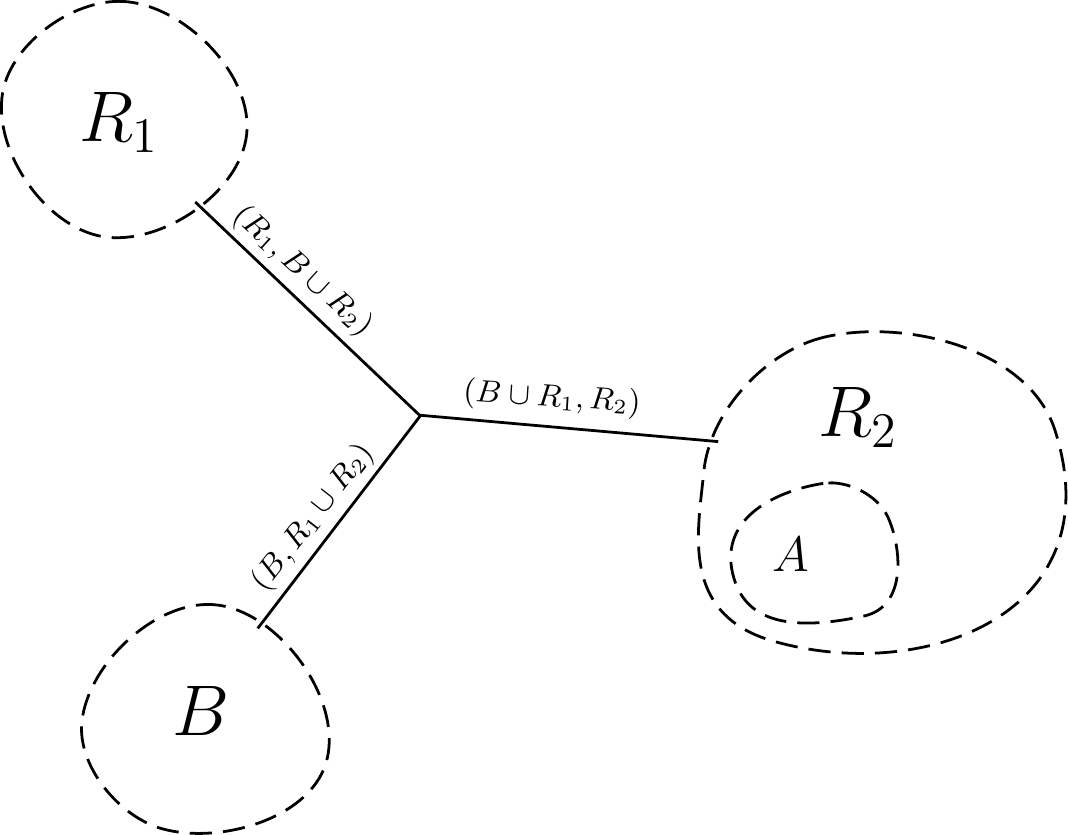}
  \caption{The two cuts incident to $\cut{B}$ in the
    decomposition.}
  \label{fig:1}
\end{figure}

From Lemma~\ref{lemma:abSubtree} we see that each vertex is incident
to at most one heavy edge. From this we deduce the following Corollary.

\begin{corollary}\label{obs:heavyEdgesFormMatching}
  For a non-trivial prime graph, its heavy edges form a matching.
\end{corollary}

\begin{lemma}\label{lemma:mmRelasjonMellomPogG}
  Let $smw(G) < k$ and let $P$ be a prime graph in a split
  decomposition of $G$ and let $A \subseteq V(P)$ with $2 \leq |A|
  \leq |V(P)| - 2$.  If no heavy edges cross the cut $(A, V(P)
  \setminus A)$ in $P$ and $\sm(A) < t$ with respect to $P$, then
  $\sm(\tot(A:P))$ with respect to $G$ is less than $9tk$.
\end{lemma}

\begin{qedproof}
  By K\"onig's theorem the size of a maximum matching and size of a minimum vertex
  cover in a bipartite graph is the same. Since $P$ is a prime graph it has no non-trivial splits and
  thus for $P$ we have $mm(V_i)=sm(V_i)$ for any $V_i \subseteq V(P)$. Thus if $\mm(A) < t$ in $P$, there must
  exist a vertex cover $C \subseteq V(P)$ for $P[A, V(P) \setminus A]$ of cardinality less than
  $t$. Based on the vertices of $C$ we create a
  vertex cover $C' \subseteq V(G)$ for the subgraph $G[{\tot(A:P)}, \compl {\tot(A:P)}]$ of $G$ having cardinality less than $9tk$,
  proving that $\mm(\tot(A:P)) < 9tk$. 
  
  We start with $C'=\emptyset$ and add for each $v \in C$ a set $C_v$ to $C'$.
  For any vertex $v \in C$ not part of any heavy pair in $P$, it follows from 
  Lemma~\ref{lemma:9kLemma} and Lemma~\ref{lemma:abSubtree} that either $\weight{v} <
  3k$ or $\weight{N(v)} < 9k$. In both cases, there is a set
  $C_v \subseteq V(G)$ of size at most $9k$ that we add to $C'$ so that each edge of $G$
  incident with $\tot(v:P)$ (and in particular those crossing the cut $\cut {\tot(A:P)}$) has an endpoint in $C'$. If on the other hand $v$ is
  part of a heavy pair $uv$ in $P$, we note, by the assumption in the lemma, that $u$ must be on the same side
  as $v$ in the cut $(A, V(P) \setminus A)$ of $P$. Again, it follows from 
  Lemma~\ref{lemma:abSubtree} that $\weight{N(\set{u,v})} < k$, so there is a set $C_v \subseteq V(G)$ of 
  size at most $k$ that we add to $C'$ so that all edges in the
  subgraph $G[{\tot(A:P)}, \compl {\tot(A:P)}]$ incident to $\tot(v:P)$ are covered by $C'$. Doing this for every vertex $v \in C$ 
  will lead to a set $C'$ with $|C'| \leq 9|C|k < 9tk$. Furthermore,
  by the definition of a vertex cover, and $\tot()$, the set $C'$
  covers all the edges of $E(G[{\tot(A:P)},{\tot(A:P)}])$.
\end{qedproof}

When deleting a vertex any cut that was a split remains a split or results in a cut with a single vertex on one side, and no new matchings are introduced.

\begin{observation}\label{obs:inducedSubgraph}
  The sm-width of a graph $G$ is at least as big as the sm-width of any
  induced subgraph of $G$.
\end{observation}

\begin{theorem}\label{thm:PrimBranch}
Let $smw(G) < k$  and let $P$ be a prime graph in a split decomposition of $G$
   We can in $\bigoh^{*}(8^k)$-time
  construct a branch decomposition $(T_P', \delta_P')$ of $P$ so that
  for each cut $(X, Y)$ of $P$ induced by an edge of $(T_P', \delta_P')$, the cut $(\tot(X:P),
  \tot(Y:P))$ of $G$ has \SM\ less than $54k^2$.
\end{theorem}

\begin{qedproof}
  If $P$ is a trivial prime graph, i.e. $|V(P) \leq 3|$, every cut $(X, Y)$ of $P$ is a
  split.  This implies by the definition of a split decomposition that
  $(\tot(X:P), \tot(Y:P))$ in $G$ also is a split of $G$. Hence,
  $\sm(\tot(X:P))$ of $G$ equals one.

  We now consider the case when $P$ is non-trivial.  Since $P$ is
  isomorphic to an induced subgraph of $G$ (this follows directly from definition of split decompositions)
  and $\smw(G) < k$, by
  Observation~\ref{obs:inducedSubgraph}, the \SMW\ of $P$ is less than
  $k$.  Also, since $P$ by definition has no splits, we
  have $\mmw(P) = \smw(P) < k$. By Theorem~\ref{theorem:mm_submod} and
  Lemma~\ref{lemma:oumSeymour}, we can compute a branch decomposition
  $(T_P, \delta_P)$ of $P$ with MM-width less than $3k$ in
  $\OumTime$-time. By a non-leaf edge of $T_P$ we mean an edge with both 
  endpoints an inner node of $T_P$. The cut in $P$ induced by a non-leaf edge of $(T_P, \delta_P)$ will
  have at least two vertices on each side. We call such cuts non-leaf cuts of $P$ induced by $(T_P, \delta_P)$.
  Note that cuts having one side containing a singleton $X=\{v\}$ are easy to deal with, 
  either the singleton is a vertex of
  $V(G)$ and then $\tot(X:P)=\{v\}$, or $v$ is a marker and the cut $(\tot(X:P),
  \tot(Y:P))$ of $G$ is a split, and thus in both cases $\sm(\tot(X:P))=1$.
  For the remainder we consider only non-leaf cuts.
  
  Denote by $h(A)$ the number of heavy edges crossing the non-leaf cut
  $(A, V(P) \setminus A)$.  If none of the non-leaf cuts of $P$
  induced by $(T_P, \delta_P)$ have heavy edges crossing them,
  i.e. $h(A)=0$ for all non-leaf cuts, we apply
  Lemma~\ref{lemma:mmRelasjonMellomPogG} with $t=3k$ and are done,
  getting for any cut $(X, Y)$ of $P$ induced by an edge of $(T_P,
  \delta_P)$ a bound of $\sm(\tot(X:P)) \leq 3k9k=27k^2$.  On the
  other hand, if some non-leaf cuts of $P$ induced by $(T_P,
  \delta_P)$ do have heavy edges crossing them, we restructure the
  decomposition $(T_P, \delta_P)$ to a decomposition $(T_P',
  \delta_P')$ as follows: for each heavy pair $a,b$ in $V(P)$ crossing
  such a non-leaf cut we remove the leaf in $T_P$ mapping to $b$ and
  make a new leaf mapping to $b$ as sibling of the leaf mapping to
  $a$.  By Corollary \ref{obs:heavyEdgesFormMatching} the heavy edges
  in $P$ form a matching, so this is easily done for all heavy edges
  of $P$ crossing non-leaf cuts, without conflicts.  Since all such
  heavy pairs are now mapped to leaves that are siblings of $T_P'$
  none of the non-leaf cuts of $P$ induced by $(T_P', \delta_P')$ will
  have a heavy edge crossing them.
  
  Let us look at how the restructuring altered the sm-value of non-leaf cuts.
  Note that for each non-leaf cut $(A', V(P) \setminus A')$ in $(T_P',
  \delta_P')$ there is an associated non-leaf cut $(A, V(P) \setminus A)$ in $(T_P, \delta_P)$ with $h(A)$ heavy edges crossing this cut, such that we move between the two cuts by moving $h(A)$ vertices across the cut.
  We have
  $\mm(A') \leq \mm(A) + h(A)$, as the maximum matching of a cut can increase by at most
  one for each vertex moved over the cut.  Moreover, by Corollary \ref{obs:heavyEdgesFormMatching} the heavy edges
  in $P$ form a matching, which means that $h(A) \leq \mm(A)$, implying
  $\mm(A') \leq 2\mm(A) \leq 2 \times 3k$.  We can therefore apply Lemma~\ref{lemma:mmRelasjonMellomPogG}
  with $t=6k$ and
  this means we have $\sm(\tot(A:P)) \leq 6k9k = 54k^2$. 
\end{qedproof}

\begin{theorem}\label{thm:compBranchDec}
  Given a  graph $G$ with $smw(G)<k$, we can compute  a branch
  decomposition  $(T, \delta)$  of $G$  of  \SMW\ less  than $54k^2$  in
  $\OumTime$-time.
\end{theorem}

\begin{qedproof} 
  For any $G'$ in a split decomposition of $G$, be it a prime graph or
  a composition of prime graphs, we claim the following: We can create
  a branch  decomposition $(T_{G'}, \delta_{G'})$ of $G'$  so that for
  each cut $\cut{A}$  in $(T_{G'}, \delta_{G'})$, the \SM\  of the cut
  $(\tot(A:G'), \tot(\compl{A}:G'))$ in $G$  is less than $54k^2$.  We
  call this latter  cut the cut induced in $G$.  We  will give a proof
  of this by induction on the number of splits in $G'$:

  In the  case that $G'$  does not  have a split,  it must be  a prime
  graph and  by Theorem \ref{thm:PrimBranch} the  hypothesis holds. If
  on the other  hand $G'$ does have a split, it  must be decomposed by
  two  subgraphs $G_1$  and $G_2$  in the  split decomposition  of $G$
  sharing a  single marker $v$.   By induction, these two  graphs have
  such branch decompositions  $(T_{G_1}, \delta_{G_1})$ and $(T_{G_2},
  \delta_{G_2})$  inducing cuts  in $G$  where the  \SM\ is  less than
  $54k^2$.   We  will  now  merge  these  two  decompositions  into  a
  decomposition  $(T_{G'}, \delta_{G'})$ for  $G'$. What  we do  is to
  identify  the   vertex  $v$  in  $T_{G_1}$  and   $v$  in  $T_{G_2}$
  ($V(T_{G'})   =  V(T_{G_1})  \cup   V(T_{G_2})$  and   $E(T_{G'})  =
  E(T_{G_1}) \cup E(T_{G_2})$). By  also joining the mapping functions
  $\delta_{G_1}$  and $\delta_{G_2}$  in  the natural  way,  we get  a
  branch decomposition $(T_{G'}, \delta_{G'})$  of $G'$ where each cut
  induced in $G$ is already a  cut induced in $G$ by either $(T_{G_1},
  \delta_{G_1})$ or $(T_{G_2}, \delta_{G_2})$. By induction $(T_{G_1},
  \delta_{G_1})$  and $(T_{G_2}, \delta_{G_2})$  only have  cuts where
  the \SM\ of all induced cuts is less than $54k^2$, so the same holds
  for $(T_{G'}, \delta_{G'})$.

  The  recursive algorithm resulting  from the  above induction  has a
  runtime of  $\OumTime$ on each prime  graph and no  more than linear
  time of work (finding the  marker) in all other partially decomposed
  graphs in  the decomposition, totalling  to a runtime  of $\OumTime$
  since  the  number of  prime  graphs  in  a split  decomposition  is
  polynomial in $n$.
\end{qedproof}

\section{Dynamic programming parameterized by \SMW} %

In this section we solve {\sc MaxCut}, {\sc Graph Coloring}, {\sc
  Hamiltonian Cycle} and {\sc Edge Dominating Set} on a graph $G$ by a
bottom-up traversal of a rooted branch decomposition $(T, \delta)$ of
$G$, in time FPT parameterized by the sm-width of
$(T,\delta)$. Previously, we did not define the tree $T$ to be rooted,
but this will help guide the algorithm by introducing parent-child
relationships. To root $T$, we first pick any edge of $T$ and
subdivide it. We then root the tree in the newly introduced vertex,
resulting in a rooted binary tree consisting of the exact same cuts as
the original decomposition.

In the bottom-up traversal of the rooted tree we encounter two
disjoint subsets of vertices $A,B \subseteq V(G)$, as leaves of two
already processed subtrees, and need to process the subtree on leaves
$A \cup B$.  There are three cuts of $G$ involved: $\cut{A}, \cut{B},
\cut{A \cup B}$, and each of them can be of type split, or of type
non-split (also called type mm for maximum-matching).  This gives six
cases that need to be considered, at least if we use the standard
framework of table-based dynamic programming. We instead use an
algorithmic framework for decision problems where we {\sc join} sets
of certificates while ensuring that the result preserves witnesses for
a 'yes' instance.  Under this framework, the algorithm for {\sc
  MaxCut} becomes particularly simple, and only two cases need to be
handled in the {\sc join}, depending on whether the 'parent cut'
$\cut{A \cup B}$ is a split or not. For the other three problems we
must distinguish between the two types of 'children cuts' in order to
achieve FPT runtime, and the algorithms are more complicated.

Let us describe the algorithmic framework. As usual, e.g. for problems
in NP, a verifier is an algorithm that given a problem instance $G$
and a certificate $c$, will verify if the instance is a
'yes'-instance, and if so we call $c$ a witness.  For our algorithms
we will use a commutative and associative function $\conc(x,y)$, that
takes two certificates $x,y$ and creates a set of certificates. This
is extended to sets of certificates $X_A,X_B$ by $\conc(X_A, X_B)$
which creates the set of certificates $\bigcup_{x_A \in X_A, x_B \in
  X_B} \conc(x_A,x_B)$.  For a graph decision problem, an input graph
$G$, and any $X \subseteq V(G)$ we define $\cert(X)$ to be %
a set of certificates on only a restricted part of $G$, which must be
subject to the following constraints:
  \begin{itemize}
  \item If $G$ is a 'yes'-instance, then $\cert(V(G))$ contains a
    witness. 
  \item For disjoint $X,Y \subseteq V(G)$ we have
    $\conc(\cert(X),\cert(Y)) = \cert(X \cup Y)$.
  \end{itemize}
  For FPT runtime we need to restrict the size of a set of
  certificates, and the following will be useful.  For $X \subseteq
  V(G)$ and certificates $x, y \in \cert(X)$, we say that $x$
  \emph{preserves} $y$ if for all $z \in \cert(\compl X)$ so that
  $\conc(y,z)$ contains a witness, the set $\conc(x,z)$ also contains
  a witness.  We denote this as $x \better_X y$.  A set $S$
  preserves $S' \subseteq \cert(X)$, denoted $S \better_X S'$, if for
  every $x' \in S'$ there exists a $x \in S$ so that $x \better_X
  x'$. A certificate $x \in \cert(X)$ so that there exists a $y \in
  \cert(\compl X)$ where $\conc(x,y)$ contains a witness, is called an
  \emph{important} certificate.

For a rooted branch decomposition $(T, \delta)$ of a graph $G$ and
vertex $v \in V(T)$, we denote by $V_v$ the set of vertices of $V(G)$
mapped by $\delta$ from the leaves of the subtree in $T$ rooted at
$v$.
With these definitions we give a generic recursive (or bottom-up)
algorithm called {\sc Recursive} that takes $(T, \delta)$ and a vertex $w$ of $T$ as input
and returns a set $S \better_{V_w} \cert(V_w)$, as follows:
\begin{itemize}
 \item at a leaf $w$ of $T$ {\sc initialize} and return the set $\cert(\{\delta(w)\})$
 \item at an inner node $w$ first call {\sc Recursive} on each of the  children nodes $a$ and $b$ and then 
 run procedure {\sc Join} on the returned input sets $S_1,S_2$ of certificates, with $S_1 \better_{V_A} \cert(V_a)$
 and $S_2 \better_{V_b} \cert(V_b)$, and return a set $S \better_{V_a \cup V_b} \conc(S_1,S_2)$
 \item at the root we will have a set of certificates $S \better_{V(G)} \cert(V(G))$ 
\end{itemize}
Calling {\sc Recursive} on the root $r$ of $T$ and running a verifier on the output
solves any graph decision problem in NP.
Correctness of this procedure follows from the definitions.
The extra time spent by the verifier is going to be
$\bigoh^*(\card{S})$, and for an FPT algorithm we will require that all
$\card{S}$ be $\bigoh^*(f(k))$, i.e.\ FPT in the sm-width $k$ of $(T, \delta)$.

In the following subsections we show how to solve each of the
respective four problems in FPT time. A rough sketch of the idea of
how this can be achieved for each of the problems is shown below. A
formal definition of each of the problems is given in each of their
respective subsections.

\paragraph{Maximum Cut.}

\textsc{MaxCut} is the one out of the four problems which has the most
simple algorithm. To compute a maximum cut, we will give an algorithm
to solve $t$-{\sc MaxCut}, which instead of maximizing a cut asks for
a cut of size at least $t$. Running $t$-{\sc MaxCut} for increasing
values of $t$, will determine the size of a maximum cut. The
certificates for this problem is subsets of vertices and a witness is
a subset $S$ so that the number of edges with one endpoint in $S$ and
one in $V(G) \setminus S$ is at least $t$ (i.e., witnesses are cuts of
size at least $t$). We show that for a cut $(A, \compl A)$ and subsets
$S_1$ and $S_2$ in $\cert(A)$, if the neighbourhood of $S_1$ and $S_2$
in $\compl A$ are the same, then one of the sets preserves the other
in $A$. As this number is bounded by $2$ and $2^{\mm(A)}$, for split
and non-split cuts, respectively, we will be able to give an FPT
algorithm for solving {\sc MaxCut}.

\paragraph{Hamiltonian Cycle.}

For {\sc Hamiltonian Cycle}, certificates are disjoint paths or
cycles, and a witness is a Hamiltonian cycle.  The important
information is what neighbourhood the endpoints of each path has over
the cut. For each certificate we keep track of the number of
\emph{path classes}, which are sets of paths with the same
neighborhood over the cut, and the size of each such path class.  The
total number of path classes over all certificates is also important.
For a split cut, the size of a class might be anything from 1 to $n$,
but there will be only one class in total. For a non-split cut of
sm-value $k$, the total number of path classes is bounded by $2^{2k}$
and since each path is vertex disjoint the number of paths in any
important certificate is bounded by $k$. Based on this the {\sc Join}
operation will be able to find a FPT-sized set of certificates
preserving a full set.

\paragraph{Chromatic Number.}

For {\sc Chromatic Number}, we will actually solve {\sc $t$-Coloring},
which asks whether the input graph can be colored by at most $t$
colors, and from this conclude that {\sc Chromatic Number} can be
solved in the same time when excluding polynomial factors. We note
that a graph of sm-width $k$, unlike graphs of treewidth $k$, may need
more than $k+1$ colors. We let all partitions into $t$ parts where the
parts induce independent sets be our certificates.  What matters for a
certificate is what kind of certificates it can be combined with to
yield a new certificate, i.e. inducing an independent set also across
the cut.  For non-split cuts, this means the number of important
certificates is bounded by the number of ways to $t$-partition the
vertices in the $k$-vertex cover of the cut, which is a function of
$k$. For a split cut, what is important is the number of parts of a
partition/certificate that have neighbors across the cut. The
certificate minimizing this number will preserve all other
certificates.  Based on this the {\sc Join} operation will be able to
find a preserving set of certificates of FPT-size.

\paragraph{Edge Dominating Set.}

For {\sc Edge Dominating Set} (or {\sc $t$-Edge Dominating Set} which
is what we actually solve) the certificates are subgraphs of $G$ and a
witness is a graph $G' = (V',E')$ so that each vertex in $V'$ is
incident with an edge in $E'$, and $E'$ is an edge dominating set of
$G$ of size at most $t$.  %
The idea of how to make an FPT {\sc Join}-procedure is that for a
vertex cover $C$ of a cut, the number of ways a certificate can
project to $C$ is limited by a function of the size of $C$. Based on
this we find a preserving set of FPT cardinality when $|C|$ is at most
$k$. When $|C|$ is not bounded by $k$, we have a split. For splits we
limit the max number of certificates needed for a preserving set by a
polynomial of $n$. This is because almost all edges on one side of the
cut affect the rest of the edges uniformly, and the other way around.

\subsection{Maximum Cut}

\subsubsection{The Problem.}

The problem \textsc{$t$-MaxCut} asks, for a graph $G$, whether there
exists a set $W \subseteq V(G)$ so that the number of edges in $G[W,
\compl W]$ is at least $t$. For a set $X$, we denote by
$\cutsize_G(X)$ the number of edges in $G[V(G) \cap X, V(G)
\setminus X]$ (note that $X$ does not need to be a subset of $V(G)$).

\subsubsection{The certificates and $\conc$.}

For \textsc{$t$-MaxCut}, we define $\cert(X)$ for $X \subseteq V(G)$
to be all the subsets of $X$, and we define $\conc(x,y)$ to be the union
function; $\conc(x,y) = \{x \cup y\}$.  
We  solve \textsc{$t$-MaxCut} by use of 
\textsc{Recursive} and the below procedure {\sc Join$_{maxcut}$} with input specification as described above. 

\subsubsection{The \textsc{Join$_{maxcut}$} function.}

\begin{tabbing}
  {\bf Procedure {\sc Join$_{maxcut}$}} \\
   xx\={\bf Output:} \=\kill
   \>{\bf Input:} \> $S_1 \better_{V_a} cert(V_a)$ and $S_2 \better_{V_b} cert(V_b)$ with  $A =V_a \cup V_b$\\
   \>{\bf Output:}\> $S \better_{A} \conc(\cert(V_a), \cert(V_b)) = \cert(A)$\\[-0.5em]
  \rule{\textwidth}{0.25pt}\\
  x\tabme{}
  \> $S' \leftarrow \set{s_1 \cup s_2 : s_1 \in S_1, s_2 \in S_2}$  /* note \emph{$S' = \conc(S_1,S_2)$} */\\
  \> $S \leftarrow \emptyset$\\
  \> $C \leftarrow$ a minimum vertex cover of $G[A, \compl A]$\\
  \>\>\>  $S $\=$\leftarrow S' \cup \{c'\}$\kill
  \>\textbf{if} $\cut{A}$ is a split \textbf{then} 
  \textbf{for} $z = 0, \ldots, n$ \textbf{do} \\
  \>\>\>  $c'$\>$\leftarrow \argmax_{c \in S'}\{\cutsize_{G[A]}(c) : 
  \card{N(\compl{A}) \cap c} = z\}$\\
  \>\>\>  $S$\>$\leftarrow S \cup \{c'\}$\\
  \>\textbf{else}
  \textbf{for} all subsets $S_{C} \subseteq C$ \textbf{do} \\  
  \>\>\>  $c'$\>$\leftarrow \argmax_{c \in S'}\left\{\cutsize_{G[A]}(c) :
    S_{C} \cap A = c \cap A \right\}$\\
  \>\>\>  $S$\>$\leftarrow S \cup \{c'\}$
\\  \>\textbf{return} $S$
\\[-0.5em]
  \rule{\textwidth}{0.25pt}
\end{tabbing}

\begin{lemma} \label{trim:cut}
  Procedure {\sc Join$_{maxcut}$} is
  correct and runs in time $\bigoh^*((\card{S_1} \card{S_2})2^k)$,
  producing a set $S$ of cardinality $\bigoh(n+2^k)$.
\end{lemma}
\begin{qedproof}
  We see that $S' \better_A \cert(A)$, since $S' = \conc(S_1, S_2)
  \better_A \cert(A)$, and $S'$ can be calculated in time
  $\bigoh^*(\card{S_1} \card{S_2})$. Finding a vertex cover of a $G[A,
  \compl A]$ can be done in polynomial time, since $G[A, \compl A]$ is a
  bipartite graph. Also, when $\cut A$ is not a split, then
  $\mm(A)\leq k$ and $\card{ C } \leq k$. Combined with a polynomial
  amount of work for each iteration of the for loops, and loops
  iterating at most $\max\{n, 2^k\}$ times (making the size of $S$
  also bounded by $n + 2^k$), the total runtime is
  $\bigoh^*(\card{S}2^k)$.

  To show that $S \better_A S'$ (and thus also $S \better_A \cert(A)$)
  we have to make sure that if there exists a witness $x$ of
  $t$-MaxCut so that for $x_A \in S'$ and $x_{\compl A} \in
  \cert(\compl A)$ we have $\{x\} \subseteq \conc(x_A,x_{\compl A})$,
  then there must exist a certificate $x' \in S$ so that $x' \better_A
  x_A$. We assume there exists such a witness $x$ with $x_A$ and
  $x_{\compl A}$ defined as above. We have two cases to consider; when
  $\scut{A}$ is a split, and when it is not.

  We first consider the case when $\cut A$ is a split. Since $x$ is a
  witness, $\cutsize_G(x) \geq t$. Let $z = \card{ N(\compl{A}) \cap x
  }$.  We have
$    \cutsize_G(x) = \cutsize_{G[A]}(x_A) + \cutsize_{G[\compl
      A]}(x_{\compl A}) + \cutsize_{G[A,\compl A]}(x), \mbox{ and }\\
    \cutsize_{G[A, \compl A]}(x) = \card{ N(\compl A) \cap x_A } \times
    \card{ N(A) \setminus x_{\compl A}} = z \card{ N(A) \setminus
      x_{\compl A}}.$
  Since $S$ contains $c \in S'$ maximizing $\max_{c \in S'} \{
  \cutsize_{G[A]}(c) : \card{N(\compl{A} \cap c)} = z\}$, we have 
 $ \cutsize_G(\conc(c,x_{\compl A})) = \cutsize_G(x) +
  \cutsize_{G[A]}(c) - \cutsize_{G[A]}(x_{A}) \geq
  \cutsize_G(x)$ meaning $\conc(c,x_{\compl A})$ is a witness, and so $S \better_A S'$.

   Now, consider the case when $\cut A$ is not a split (this means
   $\mm(A) \leq k$). Let $C$ be the vertex cover used in the procedure
   and $x_{C}$ be $x \cap C$.  As $C$ disconnects $A$ and $\compl A$,
   we have $\cutsize_G(x) = \cutsize_{G[A]}(x_A) + \cutsize_{G[\compl
     A]}(x_{\compl A}) + \cutsize_{G[A, \compl A]}(x_{C})$. We notice
   that for all $c \in S'$ so that $c \cap C = x_A \cap C$, we have
   $\cutsize_{G[A, \compl A]}(\conc(c,x_{\compl A})) = \cutsize_{G[A, \compl A]}(x_{C})$. 
Therefore, as $S$ contains the certificate
   $c$ of $S'$ where $c$ maximizes $\max_{c \in
     S'}\left\{\cutsize_{G[A]}(c) : x_A \cap {C} = c \cap A \right\}$,
   we must have that $\cutsize_G(\conc(c,x_{\compl A})) \geq
   \cutsize_G(x)$.  So $\conc(c, x_{\compl A})$ is also a witness, and
   hence $S \better S'$.
\end{qedproof}

\begin{theorem} \label{thm:maxCut}
  Given a graph $G$ and branch decomposition $(T, \delta)$ of \SMW\
  $k$, we can solve %
  \textsc{MaxCut} in time $\bigoh^*(8^k)$.
\end{theorem}

\begin{qedproof}
  In Lemma~\ref{trim:cut} we show {\sc Join$_{maxcut}$} is correct and
  produce a preserving set $S$ of size at most $\bigoh^*(2^{k})$ in
  time $\bigoh^*(|S_1||S_2|2^k)$. So, using {\sc Recursive} with {\sc
    Join$_{maxcut}$}, we know the size of both of the inputs of {\sc
    Join$_{maxcut}$} is at most the size of its output, i.e., $|S_1|,
  |S_2| \leq \bigoh^*(2^{k})$. So, each call to {\sc Recursive} has
  runtime at most $\bigoh^*(8^k)$. As there are linearly many calls to
  {\sc Recursive} and there is a polynomial time verifier for the
  certificates {\sc Recursive} produces, by the definition of
  $\better$, the total runtime for solving $t$-{\sc MaxCut}is also
  bounded by $\bigoh^*(8^k)$. To solve {\sc MaxCut}, we run the
  $t$-{\sc MaxCut} algorithm for all values of $t \leq n^2$, and hence
  have the same runtime when excluding polynomials of $n$.
\end{qedproof}
\subsection{Hamiltonian Cycle}

\subsubsection{The problem.}

For a graph $G$, a subgraph $G'$ of $G$ where $G'$ is a cycle, we say
that $G'$ is a \emph{hamiltonian cycle} of $G$ if $V(G') = V(G)$. The
decision problem \textsc{Hamiltonian Cycle} asks, for an input graph
$G$, whether there exists a hamiltonian cycle of $G$.

\subsubsection{The certificates and $\conc$.}
We notice for  a set $A \subseteq V(G)$ and  hamiltonian cycle $G'$ of
$G$  that $G'[A]$  is either  the hamiltonian  cycle itself  (if  $A =
V(G)$) or  a set of vertex  disjoint paths and  isolated vertices. For
ease  of notation,  we will  throughout this  section  regard isolated
vertices as paths (of length  zero). That is, a certificate $G'[A]$ is
always either a set of vertex disjoint paths or a cycle. Based on this
observation, it is natural to let $\cert(X)$ for $X \subseteq V(G)$ on
the problem \textsc{Hamiltonian Cycle} be all subgraphs $G'$ of $G$ so
that $V(G')  = X$  and $G'$ consists  only of  disjoint paths or  of a
cycle of  length $|V(G)|$. The witnesses of  $\cert(V(G))$ are exactly
the  certificates  that  are  hamiltonian  cycles of  $G$.  Clearly  a
polynomial  time  verifier  exists,  as  we  can  easily  confirm,  in
polynomial time,  that a  hamiltonian cycle in  fact is  a hamiltonian
cycle.  Also,  as $\cert(V(G))$  contains  all  hamiltonian cycles  of
$V(G)$, it must contain a witness if $G$ is a 'yes'-instance.

For disjoint  sets $A, B \subset  V(G)$, $G_x \in  \cert(A)$, and $G_y
\in  \cert(B)$,  we define  $\conc(G_x,G_y)$  to  be  the set  of  all
certificates $G_z  = (A \cup B,\,  E(G_x) \cup E(G_y)  \cup E')$ where
$E'$ is a subset of the edges crossing $(A,B)$. That is, $\conc(G_x,G_y)$ is
the set  of all graphs  generated by the  disjoint union of  $G_x$ and
$G_y$  and adding  edges from  $G$ with  one endpoint  in $A$  and one
endpoint  in $B$ that  are also  valid certificates.  This is  a valid
definition for $\conc$, as we  have $\cert(A \cup B) = \conc(\cert(A),
\cert(B))$.

\subsubsection{The \textsc{Join$_{HC}$} function.}
In  \textsc{Join$_{maxcut}(S_1,  S_2)$}   we  first  calculated  $S  =
\conc(S_1, S_2)$, and  later reduced the size of  $S$. However, by our
definition  of  $\conc$   for  \textsc{Hamiltonian  Cycle},  even  for
certificate  sets  $S_1,  S_2$  of  restricted  cardinality,  the  set
$\conc(S_1,S_2)$  might be  huge. Therefore,  in {\sc  Join$_{HC}$} we
cannot allow  to always run  $\conc$ inside our algorithm.  Instead we
will for each pair of certificates  in $S_1$ and $S_2$ construct a set
$S' \better  \conc(S_1, S_2)$ where  $\card{S'}$ is bounded by  an FPT
function of $n$ and $k$ while possibly $S' \subset \conc(S_1, S_2)$.

Before  we present  the  algorithm, we  need  to introduce  a few  key
observations and claims and give some new terminology.

For a  certificate $G' \in \cert(A)$  for $A \subset  V(G)$, each path
$P$ of $G'$ can be categorized by an unordered pair $(N_1, N_2)$ so
that for its two endpoints $v_1$  and $v_2$ (or single endpoint $v_1 =
v_2$ if $P$ is an isolated  vertex) we have $N_1 = N(v_1) \setminus A$
and $N_2  = N(v_2) \setminus  A$. We say  that two paths are  from the
same  \emph{class}  of paths  if  they  get  categorized by  the  same
unordered  pair.    Two  certificates  $G',  G''   \in  \cert(A)$  are
\emph{path equivalent}  if there  exists a bijection  $\sigma$ between
the paths  of $G'$ and  the paths  in $G''$ so  that for each  pair of
paths $P \in G'$ and $\sigma(P) \in  G''$, the path $P$ is in the same
path class as $\sigma(P)$.

\begin{claim}
  For  disjoint sets  $A, B  \subset V(G)$  and certificates  $G_A \in
  \cert(A)$, $G_B  \in \cert(B)$, where $G_A$ consists  of $z_A$ paths
  and $G_B$ consists of $z_B$  paths, we can compute $\conc(G_A, G_B)$
  in time $\bigoh^*(2^{4z_Az_B})$.
\end{claim}

\begin{qedproof}
  From the definition $\conc(G_A, G_B)$ contains all valid
  certificates $G'$ where $G' = G_A \cup G_B + {E'}$ where ${E'}$ consists
  of edges crossing $(A,B)$. As $G'$ must be a valid certificate, each
  vertex must have degree at most $2$, so each vertex of $A \cup B$
  incident with an edge in ${E'}$ must have degree at most $1$. As $G_A$
  and $G_B$ consist of only paths, the vertices of degree at most $1$
  are exactly the vertices that occur as an endpoint of a path in
  either $G_A$ or $G_B$.  Therefore, ${E'}$ must be a subset of the edges
  going from the at most $2z_A$ endpoints in $G_A$ to the at most
  $2z_B$ endpoints in $G_B$.  The number of such subsets is bounded by
  $2^{(2z_A2z_B)}$, and finding such a set we can do with a runtime of
  no more than a polynomial factor larger than the size of this set.
\end{qedproof}

\begin{claim}
  For  any subset  $A \subset  V(G)$  and certificates  $G_1, G_2  \in
  \cert(A)$, we have  $G_1 \better_A G_2$ if $G_1$  is path equivalent
  to $G_2$.
\end{claim}

\begin{qedproof}
  Suppose  there  is  a  certificate $G_3  \in  \cert(\compl{A})$  and
  witness $W \in \conc(G_3, G_2)$. That  means that for a set of edges
  $E_W \subseteq E(G[{A}, \compl {A}])$ we have $W = G_2 \cup G_3 + E_w$. From the
  definition  of path  classes and  path equivalence,  there  exists a
  bijection $\sigma$ from paths of $G_2$ to paths in $G_1$ so that for
  each path $P$  in $G_2$ and edges $a_1p_1, a_2,p_2  \in E_W$ so that
  $p_1,p_2$  are the  endpoints of  $P$,  there must  exist two  edges
  $a_1p_1', a_1p_2' \in E(G[A, \compl A])$ where $p_1',p_2'$ are the endpoints
  of $\sigma(P)$  in $G_1$.  Thus,  if replacing the edges  $E_W$ with
  these  edges, and replacing  each path  $P$ in  $G_2$ with  the path
  $\sigma(P)$  of  $G_3$,  we  have  an  hamiltonian  cycle.   So,  if
  $\conc(G_2, G_3)$ contains a witness, so must $\conc(G_1,G_3)$.
\end{qedproof}

\begin{lemma} \label{lemma:hc:numOfPathsAtMost_MM(A)}
  For $A \subset V(G)$, if $G'$ is an important certificate of
  $\cert(A)$, then the number of paths in $G'$ is at most $\mm(A)$.
\end{lemma}

\begin{qedproof}
  Since $G'$ is an important certificate, there must exist a
  certificate $G'' \in \cert(\compl{A})$ so that $\conc(G', G'')$
  contains a witness. This means the paths of $G'$ and $G''$ can be
  joined together by edges ${E'}$ from $G[A, \compl A]$ to form a hamiltonian
  cycle $C$.  If we direct the cycle $C$, each path of $G'$ and of
  $G''$ must be incident with exactly one in-edge and one out-edge. By
  looking at the edges in ${E'}$ going from $A$ to $\compl A$, we see
  that these edges make a matching of $G[A, \compl A]$, concluding that the
  number of paths in $G'$ is at most the size of a maximum matching,
  i.e., $mm(A)$.
\end{qedproof}

For a certificate $G' \in \cert(A)$  and path $P$ of $G'$, we say that
$P$ is an \emph{isolated} path if one of its endpoints is not incident
with an edge in $E( G[A, \compl A] )$. That is, $P$ is an isolated path if it is
categorized by a pair containing an empty set.

As each vertex of an hamiltonian cycle has degree exactly two, and for
two certificates  $G_1 \in \cert(A)$, $G_2 \in  \cert(\compl A)$, each
of  the edges in  certificate $G'  \in \conc(G_1,  G_2)$ is  either in
$E(G_1)$, $E(G_2)$, or $E( G[A, \compl A] )$, we get the following observation.

\begin{observation}\label{obs:hc:noIsolatedPaths}
  If $G'$ is an important certificate, $G'$ can not contain any
  isolated paths.
\end{observation}

When computing $\conc(p,q)$ for certificates $p$ and $q$ of $A_1$ and
$A_2$ that are both splits, we know that the number of classes of
paths in $p$ and $q$ is constant, since all paths are categorized by
the same pair. This enables us to bound the number of ways needed to
combine $p$ and $q$ in order to represent $\conc(p,q)$, since we know
a lot of them will be redundant.  When, on the other hand, one of the
two sets, for instance $A_1$, has $\mm(A_1) \leq k$, then all paths
going from $A_1$ to $\compl{A_1}$ must go through a separator of size
$ \leq k$.  Together with Observation~\ref{obs:hc:noIsolatedPaths}
this implies that each important certificate in $\cert(A_1)$ can
contain at most $k$ paths.  Using this we will again be able to reduce
the number of possible combinations of $p$ and $q$ necessary to
compute in order to get a set $S \better_A \conc(p,q)$.

\begin{center}                          
\begin{tabbing} 
\rule{\textwidth}{0.25pt}\\ 
  {\bf Procedure {\sc Join$_{HC}$}}(\=on node $w$ with children $a,b$ and $A_1=V_a, A_2 =V_b$\\
                                   \> and $A=A_1 \cup A_2$ and given $S_1 \better_{A_1} cert(A_1)$ and \\
                                   \> $S_2 \better_{A_2} cert(A_2)$)\\
 \rule{\textwidth}{0.25pt}\\ 
  x\tabme{}
  
  \> // \emph{Generating $S \better_{A} \cert(A)$}\\
  \> $S \leftarrow \emptyset$\\
  \> \textbf{for} each pair $(G_1,G_2)$ in $S_1 \times S_2$ \textbf{do}\\
  \>  \> add $G_1 \cup G_2$ to $S$\\
  \>  \> $P_1,P_2 \leftarrow $ the sets of paths in $G_1$ and $G_2$, respectively\\
  \>  \> $V_1,V_2 \leftarrow $ the sets of endpoints of $P_1$ and $P_2$, respectively\\
  \>  \> \textbf{if} $\card{P_1} > \mm(A_1)$ \textbf{or} 
                     $\card{P_2} > \mm(A_2)$ 
                     \textbf{then} \textbf{continue}\\[6pt]%
  \>  \> \textbf{if} neither $\scut {A_1}$ nor $\scut {A_2}$ is a split \textbf{then}\\
  \>  \>  \> add to $S$ the set $\{G_1 \cup G_2 + E' : E' \subseteq E(G[V_1, V_2])\}$\\[6pt]
  \> \> \textbf{else if} both of $\scut{A_1}$ and $\scut{A_2}$ are splits \textbf{then}\\
  \> \> \>\textbf{for} integers $1 \leq  z \leq z' \leq \min\set{\card{P_1},\card{P_2}}$ \textbf{do}\\
  \> \> \> \> \> $P' \leftarrow$ \=result of connecting $z'$ paths in $P_1$ and $z'$ paths in $P_2$ \\
  \> \> \> \> \>                 \>together by edges crossing $(A_1, A_2)$ to form $z$ new paths\\
  \> \> \> \> \> add to $S$ the subgraph $P_1 \cup P_2 \cup P'$\\[8pt]
  \>  \> \textbf{else} \emph{//exactly one of $\scut {A_1}$ and $\scut {A_2}$ is a split}\\
  \>  \>  \> $s \leftarrow$ either $1$ or $2$, so that $\scut{A_{s}}$ is a split\\
  \>  \>  \> $r \leftarrow$ $3 - s$ \emph{// the index opposite of $s$}\\
  \>  \>  \> remove all but $2\card{P_r}$ of $G_s$'s paths from $P_s$\\
  \>  \>  \> $V_s' \leftarrow$ the set of endpoints of the now smaller set $P_1$ of paths\\
  \>  \>  \> add to $S$ all of $\set{(G_1 \cup G_2) + E': 
    E' \subseteq E(G[V_s, V_r])}$\\[6pt]

  \> // \emph{Reducing the size of $S$}\\ 
  \> remove from $S$ all certificates that  \=are invalid,
  contain isolated paths, or \\
  \> \> contain more than $\mm(A)$ paths   \\ \tabme{}
  \> \textbf{for} $G_1,G_2 \in S$ \textbf{do}\\
  \>  \> remove $G_2$ from $S$ if $G_1$ and $G_2$ are path equivalent\\
  \> \textbf{return} $S$\\
  \rule{\textwidth}{0.25pt}
\end{tabbing}
\end{center}

\begin{lemma} \label{trim:hc} Procedure {\sc Join$_{HC}$} is correct
  and runs in time $\bigoh^*(\card{S_1}^2 \card{S_2}^2 {2}^{16k^2})$,
  producing a set $S$ of cardinality $\bigoh(n + 4^{k^2})$.
\end{lemma}

\begin{qedproof} We  will first  give a proof  of the correctness  of the
  algorithm, and  then give an analysis  of the runtime.  To prove the
  correctness, we must show  that for every important certificate $G_x
  \in \conc(S_1,  S_2)$, there  exists a certificate  $G_x' \in  S$ so
  that $G_x' \better_A G_x$, and give a bound on the size of $S$.

  As the algorithm iterates through all pairs $G_a, G_b \in S_1 \times
  S_2$, we know that for some  $G_a, G_b$, we have $G_x \in \conc(G_a,
  G_b)$. Now, let us look at the iteration of the algorithm where $G_1
  = G_a$ and $G_2 = G_b$. That is, $G_x \in \conc(G_1, G_2)$. The cuts
  $\cut  {A_1}$ and $\cut  {A_2}$ can  either both  be splits,  be one
  split and one non-split, or both be non-splits.

  The algorithm will set $V_1$ and $V_2$ to be the set of vertices in
  $G_1$ and $G_2$, respectively, of degree at most one. As each vertex
  in a (hamiltonian) cycle has degree exactly two, no vertex can have
  degree three or more in an important certificate. This means that
  for the case when neither $\scut{A_1}$ nor $\scut{A_2}$ is a split,
  when the algorithm adds $\{G_1 + G_2 + E': E' \subseteq E( G[V_1, V_2] )\}$, 
  it is clear that this set preserves $\conc(G_1, G_2)$.

  Now  suppose  exactly one  of  $\scut{A_1}$  and  $\scut{A_2}$ is  a
  split.  Without   loss  of  generality,  let   $\scut{A_2}$  be  the
  split. That is, in the algorithm we have $A_1 = A_s$ and $A_2 = A_r$
  in the  algorithm. Notice  that as each  path has two  endpoints, for
  $G_x$ to be  a valid and important certificate,  the total number of
  paths in $G_1$ where at least  one of its endpoints is incident with
  an edge in $E( G[V_1, V_2] )$  in $G_x$ is at most $2|P_2|$. Furthermore,
  as each path of $G_1$ is  in the same path class, since $\scut{A_1}$
  is  a split,  it  does  not matter  exactly  which particular  $\leq
  2|P_2|$  paths  of $G_1$  gets  incident  with  an edge  of  $E( G[V_1,V_2] )$.
  So,  when the  algorithm removes paths  from $P_1$,  since at
  least  $2|P_2|$ of  them  remain (or  all  of them,  if $|P_1|  \leq
  2|P_2|$), the set  $\{G_1 + G_2 + E' :  E' \subseteq E( G[V_1', V_2] )\}$
  where $V_2'$  is the set of  endpoints of the shrinked  set of paths
  $P_1$, preserves $\conc(G_1, G_2)$.

  For the final case, when both $\scut{A_1}$ and $\scut{A_2}$ are
  splits, we see that each path of $G_x$ can be one of at most three
  types of paths: %
  (1) $(N(A_1) \setminus A, N(A_1) \setminus A)$, %
  (2) $(N(A_2) \setminus A, N(A_2) \setminus A)$, or %
  (3) $(N(A_1) \setminus A, N(A_2) \setminus A)$.  %
  From this we can conclude that a valid certificate $G_x'$ of
  $\cert(A)$ preserves $G_x$ if $G_x'$ contains exactly the same
  number of paths as $G_x$ from each of the three path types
  mentioned. Let $\pi_1$, $\pi_2$ and $\pi_3$ be the number of paths
  in $G_x$ of type (1),(2), and (3), respectively.  The algorithm
  iterates through all integer values of $z$ and $z' \geq z$ between
  $1$ and $\min\{|P_1|,|P_2|\}$. In particular, at one iteration, $z =
  \pi_3$ and $z' = |P_1| - \pi_1 = |P_2| - \pi_2$.  Therefore,
  connecting $z'$ of the paths in $|P_1|$ and $|P_2|$ together to form
  $z = \pi_3$ new paths of type (3), we have in effect generated a
  certificate preserving $G_x$.  For the case when $\pi = 0$, the
  algorithm will not iterate through $z = 0$. However, in this case
  $G_x = G_1 \cup G_2$ which we have already added to $S$ in the very
  start of the main loop for $G_1$ and $G_2$.
  
  From this, we can conclude that the set $S$ generated before the
  last part of the algorithm, where we reduce its size, preserves
  $\conc(S_1, S_2)$. As we have already shown that path equivalent
  certificates preserve each other, the last step of reducing the size
  of $S$ is going to maintain the fact that $S \better_A \conc(S_1,
  S_2)$.
  
  At the end of the algorithm, there are no invalid certificates in
  $S$, or certificates containing isolated paths.  Also, no two
  certificates in $S$ are path equivalent, so we have the following
  regarding the size of $S$: When $\scut{A}$ is a split, there are at
  most $n$ path classes, so the size of $S$ is bounded by $n$. When
  $\scut{A}$ is not a split, we know each certificate contains at most
  $\mm(A)$ paths.  Furthermore, as the minimum vertex cover of a
  bipartite graph is of the same size as the maximum matching of the
  same graph, there is a vertex cover $C$ of $G[A, \compl A]$ of size
  at most $\mm(A)$.  This means that each neighbourhood over $A$
  ($N(S) \setminus A$ for some set $S \subseteq A$) is one of at most
  $2^{\mm(A)}$ possibilities.  This means that each path can be
  represented by one of at most $(2^{\mm(A)})^2$ pairs of
  neighbourhoods. As each certificate contain at most $\mm(A)$ paths,
  we can conclude that the cardinality of $S$ is no more than
  $((2^{\mm(A)})^2)^{\mm(A)} \leq 4^{k^2}$.

  For the runtime of the algorithm, we notice that the first main loop
  runs at most $\card{S_1} \times \card{S_2}$ times. For each
  iteration we do at most $\bigoh(2^{8k^2})$ operations (worst case is
  when exactly one of $\scut{A_1}$ and $\scut{A_2}$ is a split -- then
  $|V_s'| \leq 4k$ and $|V_r| \leq 2k$).  So, we do at most
  $\bigoh^*((\card{S_1}\cdot\card{S_2})2^{8k^2})$ operations on
  the first part. For the latter part, the size of $S$ is bounded by
  the runtime of creating it, i.e., $S$ is bounded by
  $\bigoh^*((\card{S_1}\cdot\card{S_2})2^{8k^2})$. So, as the
  algorithm in the end iterates through all pairs of elements in $S$,
  the final runtime is
  $\bigoh^*((\card{S_1}^2\card{S_2}^2)2^{16k^2})$.
\end{qedproof}

\begin{theorem} \label{thm:hc}
  Given a graph $G$ and branch decomposition $(T, \delta)$ of \SMW\
  $k$, we can solve %
  \textsc{Hamiltonian Cycle} in time $\bigoh^*(2^{24k^2})$.
\end{theorem}
\begin{qedproof}
  In Lemma~\ref{trim:hc} we show the procedure {\sc Join$_{HC}$} is
  correct and runs in time $\bigoh^*(\card{S_1}^2 \card{S_2}^2
  {2}^{16k^2})$, producing a set $S$ of cardinality $\bigoh(n +
  4^{k^2})$. So, using {\sc Recursive} with {\sc Join$_{HC}$}, we know
  the size of both of the inputs of {\sc Join$_{HC}$} is at most the
  size of its output, i.e., $|S_1|, |S_2| \leq \bigoh^*(4^{k^2})$. So,
  each call to {\sc Recursive} has runtime at most
  $\bigoh^*(2^{24k^2})$. As there are linearly many calls to {\sc
    Recursive} and there is a polynomial time verifier for the
  certificates {\sc Recursive} produces, by the definition of
  $\better$, the total runtime is also bounded by
  $\bigoh^*(2^{24k^2})$.
\end{qedproof}

\subsection{$t$-Coloring}
\subsubsection{The problem.}

The decision problem \textsc{$t$-Coloring} asks for an input graph
$G$, whether there exists a labelling function $c$ of the vertices of
$V(G)$ using only $t$ colors in such a way that no edge has its
endpoints labelled with the same color. Equivalently, it asks whether
there exists a $t$-partitioning of the vertices so that each part
induces an independent set. For simplicity, we will allow empty parts
in a partition (e.g., $\{\{x_1,x_2,x_3\}, \{x_4\}, \emptyset,
\emptyset, \emptyset\}$ is a $5$-partition of $\{x_1,x_2,x_3,x_4\}$).

For a set $A$ and partition $p = \{p_1, p_2, \ldots, p_j\}$, we denote
by $A  \cap p$ the  partition $\{p_1 \cap  A, p_2 \cap A,  \ldots, p_j
\cap A\}$.

\subsubsection{The certificates and $\conc$.}

For \textsc{$t$-Coloring}, we define $\cert(X)$ for $X \subseteq V(G)$
to be all $t$-partitions of $X$ and $\conc(p,q)$ for
$p=p_1,\ldots,p_t$ and $q=q_1,\ldots,q_t$ to be the following set of
partitions %
\[%
  \set{ \bigcup_{i \leq t} \set{  \set{p_{i} \cup q_{\sigma(i)}}  } 
        : 
        \sigma \mbox{ is a permutation of } \set{1,\ldots, t} 
  } \enspace.
\]%
This satisfies the constraint $\conc(\cert(X),\cert(Y)) = \cert(X \cup
Y)$, so it is a valid definition of $\conc$.

We can easily construct a polynomial time algorithm that given a
$t$-partition $p$ of independent sets (which there must exist at least
one of if $G$ is a 'yes'-instance) is able to confirm that $G$ is a
'yes'-instance. So $t$-partitions of $V(G)$ forming independent sets
will be our witnesses.

\subsubsection{The \textsc{Join$_{col}$} function.}

The main observation for the design of this procedure is that whenever
$\cut A$ is a split, there exists a single element $x \in \cert(A)$ so that $\set{x} \better_{A} \cert(A)$.  Also,
when $\mm(A) < k$, there is a separator of $A$ and $\compl A$ that
intersect by less than $k$ parts of any witness.
In the procedure {\sc Trim$_{\textsc{col}}$} below we use this to trim the number of certificates to store at each step
of the algorithm; if two certificates ``projected'' to the separator
of $A$ and $\compl A$ is the same partition, we only store one of
them.
This results in less than $k^k$ certificates to store.

For two $t$-partitions  $P$, $Q$, we say that  we \emph{merge} $P$ and
$Q$ when  we generate  a new partition  $R$ by pairwise  combining the
parts (by union) of  $P$ with the parts of $Q$ in  such a way that $R$
is a $t$-partition  where each part is an independent  set. If $P$ and
$Q$ can be merged, we say that $P$ and $Q$ are \emph{mergeable}.

\begin{lemma}\label{lemma:col:merging}
  Given two $t$-partitions  $P$ and $Q$, deciding whether  $P$ and $Q$
  are mergeable,  and merging $P$ and $Q$  if they are, can  be done in
  polynomial time.
\end{lemma}

\begin{qedproof}
  We can  check whether  $P$ and  $Q$ are mergeable  by reducing  it to
  deciding  whether  a bipartite  graph  has  a  perfect matching:  We
  generate a bipartite graph $B = (P,Q)$ where each vertex/part $p \in
  P$  is adjacent  to $q  \in  Q$ if  and only  if  $p \cup  q$ is  an
  independent  set.  When  $P$  and  $Q$ partition  sets  that do  not
  intersect, then  we can merge  $P$ and $Q$  by for each edge  in the
  matching, combine the respective parts the edge is incident with. If
  there is no perfect matching in $B$, then that must mean there is no
  way of  pairwise combining the  parts of $P$  with the parts  $Q$ so
  that all combined parts are independent.

  If $P$ contain parts that share  elements with parts of $Q$, then we
  know these  parts must be combined  in all merged  partitions, so we
  combine all these  sets and then run the  above reduction to perfect
  matching (on the parts that  do not have intersecting elements).  If
  a part  intersect with more  than one other  part, then $P$  and $Q$
  cannot be merged.
\end{qedproof}

\begin{lemma}  \label{lemma:col:separator}  Let  $A$  be a  subset  of
  $V(G)$, $C  \subseteq V(G)$ a separator  of $A$ and  $\compl A$, and
  $P_A$ and $P_C$ be $t$-partitions  of $A$ and $C$, respectively.  If
  $P_C$ and $P_A$ merge to a $t$-partition $P_A'$, then for any set $B
  \subseteq (C \cup \compl A)$ and $t$-partition $P_B$ of $B$, $P_B$ is
  mergeable with $P_A'$ if and only if $P_C$ is mergeable with $P_B$.
\end{lemma}

\begin{qedproof}
  For the forward direction, we notice that for any $t$-partition $P'$
  resulting from  merging $P_A'$ with $P_C$, the  $t$-partition $P^* =
  P' \cap (C \cup  B)$ is a result of merging $P_A'  \cap (C \cup B) =
  P_C$ with $P_B \cap (C \cup B) = P_B$, and hence $P_C$ and $P_B$ are
  mergeable.

  We will  prove the backwards  direction by constructing  a partition
  $P'$ which  is the  result of merging  $P_A'$ with $P_B$.   As $P_C$
  partitions exactly the set $C$, and both $P_A$ merge with $P_C$ to a
  $t$-partition $P_A'$  and $P_B$ merge with $P_C$  to a $t$-partition
  $P_B'$, there  is an ordering of  the parts of $P_A'$  and $P_B'$ so
  that  the $i$-th  part of  $P_A'$  intersected with  $C$ equals  the
  $i$-th  part of  $P_B'$ intersected  with $C$.  We let  $P'$  be the
  multiset resulting from pairwise combining the $i$-th part of $P_A'$
  with the $i$-th part of $P_B'$.  As the only vertices that occur in
  parts of  both $P_A'$ and of $P_B'$  is the set of  vertices $C$, by
  combining the parts based on the ordering we described, we have made
  sure that each  vertex appear in exactly on part  of $P'$, and hence
  $P'$ is a $t$-partition.
  
  To show that  each part in $P'$ is an independent  set, we assume by
  contradiction that two adjacent vertices $x$ and $y$ are in the same
  part  of  $P'$. By  how  we constructed  $P'$,  no  two vertices  in
  different parts in either $P_A'$ or  $P_B'$ will be in the same part
  in $P'$. Therefore, as $P_A'$  partitions $A \cup C$ into parts that
  are independent  sets, and $P_B'$ does  the same for $B  \cup C$, if
  $x$ and $y$ are  adjacent and in the same part, one  of them must be
  in $A  \setminus C$ and the  other in $B \setminus  C$. However, $C$
  disconnects $A$ and  $B$, and hence no edge  exists between vertices
  in $A \setminus  C$ and $B \setminus C$,  contradicting that $x$ and
  $y$ are adjacent.
\end{qedproof} 

\begin{tabbing} 
  \rule{\textwidth}{0.25pt}\\ 
  {\bf Procedure {\sc Trim$_{\textsc{col}}$}} (on input $S \subseteq \cert(A)$, with $A \subseteq V(G)$)\\
  \rule{\textwidth}{0.25pt}\\ x\tabme{}

  \> remove from $S$ the partitions containing parts that are not independent sets\\ 
  \> \textbf{if} $\cut A$ is a split \textbf{then}\\
  \>  \> mark one certificate $P' \in S$ where \\
  \>  \> \> \>the number of non-empty parts in $N(\compl A) \cap P'$ is minimized\\
  \> \textbf{else} \\
  \>  \> $C \leftarrow $ minimum vertex cover of $A$\\
  \>  \> \textbf{for} each $t$-partition $P_C$ of $C$ \textbf{do}\\
  \>  \> \> mark one certificate $P_C' \in S$ which is mergeable with $P_C$\\
  \> \textbf{return} the marked certificates of $S$\\
  \rule{\textwidth}{0.25pt}\\ 
\end{tabbing}

\begin{lemma} \label{lemma:col:trim} The procedure
  \textsc{Trim$_{\textsc{col}}$} with input $S \subseteq A$ for $A
  \subseteq V(G)$ returns a set $S' \better_{A} S$ of size at most $\sm(A)^{\sm(A)}$, and has a runtime of
  $\bigoh^*(\card{S}\sm(A)^{\sm(A)}).$
\end{lemma}

\begin{qedproof}
  The  runtime  is  correct as  a  result  of  the fact  that  merging
  (Lemma~\ref{lemma:col:merging})  takes polynomial  time  to execute,
  and it produces a set $S' \subseteq S$ of cardinality at most $1 = \sm(A)^{\sm(A)}$ if
  $\scut{A}$ is  a split, and  cardinality at most  $\mm(A)^{\mm(A)} =
  \sm(A)^{\sm(A)}$ otherwise. We now show that $S' \better_{A} S$:

  For contradiction, let us assume there exists a certificate $P_w \in
  \cert(\compl A)$ so that for a certificate $P \in S$ the set $\conc(P_w, P)$
  contains a  witness, while $\conc(\set{P_w},  S')$ does not  contain a
  witness. Let us first assume $\cut A$ is a split:
  
  As $\cut A$ is a split, each part of $P$ is either adjacent to all
  of $N(A)$, or not adjacent to $\compl A$ at all. Let $z$ be the
  number of parts in $P$ that are adjacent to all of $N(A)$. This also
  means $z$ is a lower bound to the number of parts in $P_w$ that do
  not have neighbours in $A$.  In \textsc{Trim$_{\textsc{col}}$} we
  mark (and thus output) one certificate $P'$ where at most $z$ parts
  have neighbours $\compl A$. Thus for each of the at most $z$ parts
  in $P'$ that are adjacent to $A$, there is a part in $P_w$ not
  adjacent to any vertex in $A$. So, we can conclude that $P'$ and
  $P_w$ can be merged together to form a witness.  This contradicts
  that $\conc(\set{P_w}, S')$ does not contain a witness when $\cut A$
  is a split.

  Now, let us assume $\cut A$ is not a split.  Let $W$ be a witness in
  $\conc(\set{P_w},  S)$.   For the  vertex  cover  $C$,  we have  the
  following smaller $t$-partitions of $W$; $P_C = W \cap C$ and $P_W =
  W  \cap C$. By  Lemma~\ref{lemma:col:separator}, as  $C$ disconnects
  $A$ from  the rest of the  graph, any $t$-partition  of $A$ mergeable
  with $P_C$ is  mergeable with $P_W$.  The algorithm  assures that for
  all $t$-partitions $P_C$ of $C$,  whenever a $t$-partition in $S$ is
  mergeable with $P_C$, at  least one $t$-partition mergeable with $P_C$
  exists  in $S'$. From  this we  can conclude  that $\conc{\set{P_w},
    S'}$ preserves $\conc(\set{P_w}, S)$.
\end{qedproof}

\begin{center}                          
\begin{tabbing} 
\rule{\textwidth}{0.25pt}\\ 
{\bf Procedure {\sc Join$_{col}$}}(  \= on node $w$ with children $a,b$ and $A_1=V_a, A_2 =V_b$\\
                                    \> and $A=A_1 \cup A_2$ and given $S_1 \better_{A_1} cert(A_1)$ and\\
                                    \> $S_2 \better_{A_2} cert(A_2)$)\\
 \rule{\textwidth}{0.25pt}\\ 
    x\tabme{}

  \>  $S \leftarrow \emptyset$\\
  \>  $C_{1}, C_{2} \leftarrow$ minimum vertex cover of $G[{A_1}, \compl {A_1}]$ and $G[{A_2}, \compl {A_2}]$, respectively\\[4pt]    
  \>  \textbf{for} $P_1 \in S_1$ and $P_2 \in S_2$ \textbf{do}\\
  \>\>  $N_1 \leftarrow$ parts of $P_1$ not adjacent to $N(A_1)$\\
  \>\>  $N_2 \leftarrow$ parts of $P_2$ not adjacent to $N(A_2)$\\[6pt]

  \>\>  \textbf{if} both $\scut{A_1}$ and $\scut{A_2}$ are splits  \textbf{then}\\ 
  \>\>\>  \textbf{for} $0 \leq z \leq \min\set{\card{N_1},\card{N_2}}$ \textbf{do} \\
  \>\>\>\>  add to $S$ a $t$-partition generated (if possible) by \\
  \>\>\>\>\>\>\>       merging the two partitions $P_1$ and $P_2$ together \\
  \>\>\>\>\>\>\>       in such a way that exactly $z$ of the parts of $N_1$  \\
  \>\>\>\>\>\>\>       is merged with parts of $N_2$\\[6pt]

  \>\>  \textbf{else if} neither $\scut{A_1}$ nor $\scut{A_2}$ is a split \textbf{then}\\ 
  \>\>\>  \textbf{for} each $t$-partition $P_c$ of $C_1 \cup C_2$ \textbf{do}\\
  \>\>\>\>  \textbf{if} both $P_1$ and $P_2$ are mergeable with $P_c$ \textbf{then}\\
  \>\>\>\>\>  $P_1^c \leftarrow $ merge $P_c$ and $P_1$\\
  \>\>\>\>\>  $P' \leftarrow$ merge $P_1^c$ and $P_2$\\
  \>\>\>\>\>  add $P' \cap (A \cup B)$ to $S$\\[6pt]

  \>\>  \textbf{else}\\
  \>\>\>  assign $s,r \in \set{1,2}$ so that $\scut{A_s}$ is the split and $r \not= s$\\ 
  \>\>\>  \textbf{for} each $t$-partition $P_c$ of $C_r$ mergeable with $P_r$ \textbf{do}\\
  \>\>\>\>  \textbf{for} each subset $Q$ of the non-empty parts of $P_c$ \textbf{do}\\
  \>\>\>\>\>  $P_q \leftarrow$   \=  merge (if possible) $P_c$ with $P_s$ so that $Q$ is exactly the \\
  \>\>\>\>\>                     \>  non-empty parts of $P_c$ that get combined with $N_s$\\
  \>\>\>\>\>  $P' \leftarrow $ the $t$-partition generated by merging $P_q$ with $P_s$\\
  \>\>\>\>\>  add to $S$ the $t$-partition $P' \cap (A \cup B)$\\[6pt]

  \>  \textbf{return} \textsc{Trim$_{col}$}($S$)\\
 \rule{\textwidth}{0.25pt}\\
\end{tabbing}
\end{center}

\begin{lemma}  \label{lemma:col:join} Procedure {\sc  Join$_{col}$} is
  correct  and runs in  time $\bigoh^*(\card{S_1}\card{S_2}{k^{3k}})$,
  producing a set $S$ of cardinality $\bigoh(k^k)$.
\end{lemma}

\begin{qedproof}
  We will go through the three  cases in the algorithm (when zero, one
  or  two  out of  the  two  cuts  $\scut{A_1}$ and  $\scut{A_2}$  are
  splits), and show that for each pair $(P_1, P_2)$ of certificates in
  $S_1 \times S_2$, the  output of the algorithm preserves $\conc(P_1,
  P_2)$. As  the last  line of the  algorithm assures that  the output
  preserves $S$ (by Lemma~\ref{lemma:col:trim}) and the cardinality of
  the output  is of the  correct size, we  only need to show  that $S$
  preserves $\conc(P_1, P_2)$ and that the runtime is correct.
  
  Suppose neither of  the two cuts are splits.  In this case, whenever
  there  exists a  certificate of  $\conc(P_1, P_2)$  mergeable  with a
  $t$-partition $P_C$ of  the separator $C = C_1 \cup  C_2$ of $A$ and
  $\compl A$, the algorithm assures that there is going to be at least
  one   $t$-partition  of  $A$   in  $S$   mergeable  with   $P_C$.  By
  Lemma~\ref{lemma:col:separator}, this means $S$ is going to preserve
  $\conc(P_1, P_2)$.

  If both of the two cuts are splits, then for each part $p_i$ of
  every $t$-partition of $A$ the neighbourhood $N(p_i) \setminus A$,
  must either be empty, $N(A_1)$, $N(A_2)$, or $N(A)$. In this case
  the algorithm generates, for each possible $z$ so that there exists
  an important certificate $P^* \in \conc(P_1, P_2)$ where the number
  of parts with empty neighbourhoods is exactly $z$ (and thus, the
  number of parts with neighbourhood equal $N(A_1)$, $N(A_2)$ and
  $N(A)$, is $\card{N_2} - z$, $\card{N_1} - z$, and $t -
  \card{N_1}+\card{N_2} + z$, respectively), generates a $t$-partition
  of independent sets where there the same number of parts with each
  of the particular four neighbourhoods is equal to the number of
  parts of each neighbourhood for $P^*$.  We can observe that when two
  $t$-partitions over the same set both consist of only independent
  sets, and there is a correspondence between the parts of both
  partitions so that the neighbourhood in $\compl{A}$ of each pair of
  corresponding parts is the same, the two partitions are mergeable to
  exactly the same $t$-partitions of $\compl{A}$.  Therefore, the set
  of $t$-partitions the algorithm generates in the case when both cuts
  are splits, preserves $\conc(P_1, P_2)$.

  If exactly one of the cuts are splits, let us assume without loss of
  generality that $\scut{A_1}$ is the split. As above, each part of
  $P_1$ is one of two types; adjacent to $N(A_1)$, or not adjacent to
  $N(A_1)$.  When $C$ is a vertex cover of $G[{A_2}, \compl {A_2}]$, for each
  important certificate $P^* \in \conc(P_1, P_2)$, we have a
  $t$-partition $P_C = P^* \cap C$ of $C$. Of the parts in $P_c$, some
  subset $Z$ of the non-empty parts get combined with the parts of
  $P_1$ that do not have any neighbours in $N(A_1)$ and the rest of
  the non-empty parts get combined with the other type of parts in
  $P_1$. Since $C$ disconnects $A_2$ from the rest of the graph, there
  must be a bijection from each part in $P^*$ to parts with the exact
  same neighbourhood in $N(A)$ for the $t$-partition generated by the
  algorithm in the last if/else case, for $Q = Z$. This, in turn,
  means that the two partitions preserve each other, and so we can
  conclude that the set of certificates the algorithm produces
  preserves $\conc(P_1, P_2)$.

  The runtime of the algorithm we get as follows: Excluding
  polynomials of $n$, we get the worst case runtime when neither $A_1$
  nor $A_2$ are splits. Then $S$ grows to be as large as $k^{2k}$ (the
  number of $t$-partitions of $C_1 \cup C_2$) for each pair of
  certificates in $S_1 \times S_2$.  This implies a runtime of
  $\bigoh^*(\card{S_1}\card{S_2}k^{2k})$ for the entire part before
  the execution of \textsc{Trim$_{\textsc{col}}$}, which by
  Lemma~\ref{lemma:col:trim} takes
  $\bigoh^*(\card{S_1}\card{S_2}k^{3k})$-time given the size of $S$.
\end{qedproof}

\begin{theorem} \label{thm:col} Given a graph $G$ and branch
  decomposition $(T, \delta)$ of \SMW\ $k$, we can solve %
  {\sc Chromatic Number} in time $\bigoh^*(k^{5k})$.
\end{theorem}
\begin{qedproof}
  In Lemma~\ref{lemma:col:join} we show {\sc Join$_{col}$} is correct
  and runs in time $\bigoh^*(\card{S_1}\card{S_2}{k^{3k}})$, producing
  a set $S$ of cardinality $\bigoh(k^k)$. So, using {\sc
    Recursive} with {\sc Join$_{col}$}, we know the size of both of
  the inputs of {\sc Join$_{col}$} is at most the size of its output,
  i.e., $|S_1|, |S_2| \leq \bigoh^*(k^{k})$. So, each call to {\sc
    Recursive} has runtime at most $\bigoh^*(k^{5k})$. As there are
  linearly many calls to {\sc Recursive} and there is a polynomial
  time verifier for the certificates {\sc Recursive} produces, by the
  definition of $\better$, the total runtime is also bounded by
  $\bigoh^*(k^{5k})$. To solve {\sc Chromatic Number} when we have an
  algorithm for $t$-{\sc Coloring}, we simply run the $t$-{\sc
    Coloring} algorithm for each value of $t \leq n$, giving the same
  runtime for {\sc Chromatic Number} when excluding polynomial factors
  of $n$.
\end{qedproof}

\subsection{Edge Dominating Set} 

\subsubsection{The problem.}                       
The  decision problem  \textsc{$t$-Edge  Dominating Set}  asks for  an
input graph  $G$, whether  there exists a  set $E' \subseteq  E(G)$ of
cardinality at most $t$, so that  for each edge $e \in E(G)$ either $e
\in E'$  or $e$ shares an  endpoint with an  edge $e' \in E'$.  We say
that an edge $e' \in E'$ \emph{dominates} $e \in E(G)$ if $e$ and $e'$
share an endpoint.
                    
\subsubsection{The certificates and $\conc$.}
For  \textsc{$t$-Edge Dominating  Set},  we define  $\cert(X)$ for  $X
\subseteq V(G)$ to be all subgraphs of $G[X]$. This might seem odd, as
we are looking for a set  of edges. However, for a certificate $G' \in
\cert(X)$, the set $V(G')$ is  going to make the algorithm simpler.  A
witness  will  be a  subgraph  $G_w$ such  that  $E(G_w)$  is an  edge
dominating set  of size  at most  $t$ and each  vertex of  $V(G_w)$ is
incident with an  edge of $E(G_w)$. Checking the  latter can obviously
be  done in  polynomial time  and checking  that $E(G_w)$  is  an edge
dominating set of size at most $t$ can also be done in polynomial time
since {\sc $t$-Edge Dominating Set} is in NP.

For  disjoint sets  $A, B  \subseteq V(G)$  and certificates  $G_A \in
\cert(A)$ and  $G_B \in \cert(B)$,  we define $\conc(G_A, G_B)$  to be
the set %

\[
\set{G_A+G_B+E' : E' \subseteq E( G[V(G_A), V(G_B)] )}.
\] %

We  can  see  that  this  definition  satisfies  $\cert(A  \cup  B)  =
\conc(\cert(A), \cert(B))$.

\subsubsection{The \textsc{Join$_{\textsc{eds}}$} function.}

Given a  graph $G'$ and vertex  $v \in V(G')$,  we say that $v$  is an
\emph{isolated} vertex of $G'$ if it  is not incident with any edge of
$E(G')$. If a  vertex is not isolated, it  is \emph{non-isolated}.  We
say  a set  of edges  $E'$ \emph{span}  a set  of vertices  $X$  if $X
\subseteq V(E')$. We denote by $I(G')$ the isolated vertices of $G'$.

We say that a certificate  $G' \in \cert(A)$ is \emph{locally correct}
if all edges in $G[A]$ have an endpoint in $V(G')$ and all the isolated
vertices  of   $V(G')$  are  in  $N(\compl  A)$.    A  certificate  is
\emph{locally incorrect} if it is  not locally correct.  We see that a
certificate in $\cert(A)$ which is locally incorrect cannot also be an
important certificate as  there exists an edge in  $G[A]$ which is not
dominated  by edges in  $E(G')$ and  cannot be  dominated by  edges in
$G[A, \compl A]$.

We observe  that for  two locally correct  certificates $G_1,  G_2 \in
\cert(A)$, if  $N(\compl A) \cap G_1  = N(\compl A) \cap  G_2$ and the
isolated vertices of  $G_1$ equal those of $G_2$,  then $G_1 \better_A
G_2$ if $\card{E(G_1)} \leq \card{E(G_2)}$.

\begin{lemma}  \label{lemma:EDS_matching_argument} Given  a  graph $G$
  without isolated vertices and a subset $A \subseteq V(G)$, a minimum
  cardinality set  $X \subseteq E(G)$  of edges spanning  the vertices
  $A$ can be found in polynomial time.
\end{lemma}

\begin{qedproof}
  Let $M$ be  a maximum matching of $G[A]$. Any  set spanning $A$ must
  be of size at  least $|M| + (|A| - 2|M|) =  |A| - |M|$, as otherwise
  there must  exist a matching in  $G[A]$ larger than $M$.  Let $R$ be
  the set of vertices in $A$ not incident with any edge in $M$.  As no
  vertex of  $G$ is isolated, each  vertex in $R$ is  incident with at
  least one edge in $G$. Thus, we can easily find a set $E_R$ of $|R|$
  edges spanning  $R$. We claim  that the  set $X =  M \cup E_R$  is a
  minimum  cardinality set of  edges spanning  $A$; Clearly  $X$ spans
  $A$, as each  vertex not spanned by $M$ is  by definition spanned by
  $E_R$. Furthermore, the size of $R$  is exactly $|A| - 2|M|$, so the
  size of $X$ is $|M| + (|A| -  2|M|) = |A| - |M|$. Hence, the set $X$
  is a minimum cardinality set spanning $A$.
\end{qedproof}

As usual, our join-procedure will consist of a part where we join
together pairs of certificates, imitating $\conc$, and a part where
the size of the output is reduced using a trim-procedure that ensures
the output is preserving. For the \textsc{Edge Dominating
  Set}-problem, the trimming part consists of two procedures; one for
when the cut in question is a split, and one for when it is not a
split. We first present the simplest of the two, namely the one used
when the cut is a split (\textsc{Trim$_{\textsc{eds}}$-split}).

\begin{center}                          
\begin{minipage}{1.0\linewidth}
\begin{tabbing} 
  \rule{\textwidth}{0.25pt}\\ 
  {\bf Procedure {\sc Trim$_{\textsc{eds}}$-split}} 
  (on $S \subseteq \cert(A)$ where $\scut A$ is a split)\\
  \rule{\textwidth}{0.25pt}\\ 
  x\tabme{}
  
  \>  remove from $S$ all locally incorrect certificates\\
  \>  \textbf{for} each $0 \leq x_1,x_2 \leq n$ \textbf{do}\\
  \>\>  mark one certificate $G' \in S$ of 
        minimized $\card{E(G')}$ where\\
  \>\>\>  $x_1$ equals $\card{I(G')}$, and\\
  \>\>\>  $x_2$ equals $\card{V(G') \cap N(\compl A))}$\\
  \>  return all the marked certificates in $S$\\

 \rule{\textwidth}{0.25pt}
\end{tabbing}
\end{minipage}
\end{center}

\begin{lemma} \label{lemma:eds:trim:split} The procedure
  \textsc{Trim$_\textsc{eds}$-split} on $S \subseteq \cert(A)$ for $A
  \subseteq V(G)$ where $\scut A$ is a split, returns a set $S'
  \subseteq S$ so that $S' \better_A S$ and $\card{S'} \in
  \bigoh(n^2)$ in $\bigoh^*(\card{S})$-time.
\end{lemma}

\begin{qedproof}
  The  algorithm clearly does  $(n+1)^2$ number  of iterations  in the
  for-loop, and for  each of these iterations it  iterates through the
  list $S$.  For each element in  $S$, it does a  polynomial amount of
  work, so the total  runtime is $\bigoh^*(|S|)$. Furthermore, at most
  one element is marked to be  put in the output at each iteration, so
  the output of the algorithm is of size $\bigoh(n^2)$.

  For  the correctness  of  the  algorithm, notice  that  for any  two
  locally correct certificates $G_1, G_2 \in \cert(A)$, if $|I(G_1)| =
  |I(G_2)|$  and $|V(G_1) \cap  N(\compl A)|  = |V(G_2)  \cap N(\compl
  A)|$, then  $G_1$s preserves  $G_2$. This is  because each  vertex in
  $N(\compl A)$  is adjacent to  exactly the same vertices  of $\compl
  A$, and  so for any set $X_2  \subseteq N(\compl A)$, if  there is a
  set  of edges  spanning  $X \subset  \compl{A}$  and $X_1  \subseteq
  N(\compl A)$, there  is also a set of edges  of the same cardinality
  spanning $X_2$ and $X$ as long as $\card{X_2} = \card{X_1}$.
\end{qedproof}

Now we describe the trim-procedure for when the respective cut is not
a split ({\sc Trim$_{\textsc{eds}}$-non-split}).  This procedure is
more complicated than {\sc Trim$_{\textsc{eds}}$-split}. The core
idea of the procedure is that when $G_a + G_{\compl a} + E_w \in \conc(G_a,
G_{\compl a})$ is a witness, then also for any locally correct certificate
$G_i$, the certificate $G_i + G_{\compl a} + E_w' \in \conc(G_i, G_{\compl a})$ is
also a witness, as long as $V(G_i) \cup V(G_{\compl a})$ is a vertex cover of
$G[A, \compl A]$, $E_w'$ spans $I(G_i)$ and the vertices in $V(E_w)
\cap \compl A$, and $|E(G_i)|+|E_w'| \leq |E(G_a)|+|E_w|$.

\begin{center}                          
\begin{minipage}{1.0\linewidth}
\begin{tabbing} 
  \rule{\textwidth}{0.25pt}\\ 
  {\bf Procedure {\sc Trim$_{\textsc{eds}}$-non-split}} 
  (on $S \subseteq \cert(A)$)\\
  \rule{\textwidth}{0.25pt}\\ 
  x\tabme{}
  
  \> remove from $S$ all locally incorrect certificates\\
  \> $C \leftarrow $ minimum vertex cover of $G[A, \compl A]$\\
  \> \textbf{for} all $Q \subseteq R \subseteq C$ \textbf{do}\\
  \> \> $R_a \leftarrow R \cap A$\\
  \> \> $R_{\compl a} \leftarrow R \cap \compl A$\\
  \> \> mark one certificate $G_i \in S$ 
        minimizing $\card{E_{i}} + \card{E(G_i)}$ where:\\
  \> \> \> $V(G_i) \cap C = R_a$,\\
  \> \> \> $C \setminus (R_{\compl a} \cup A) \subseteq N(V(G_i))$, and\\
  \> \> \> $E_{i} \leftarrow$ a minimum subset of 
           $E( G[V(G_i), R_{\compl a}])$ spanning 
           $(I(G_i) \cup R_{\compl a}) \setminus Q$\\
  \> \> (if no such $G_i$ exists, 
         then don't mark any certificate)\\
  \> return the set of all the marked certificates in $S$\\

 \rule{\textwidth}{0.25pt}
\end{tabbing}
\end{minipage}
\end{center}

\begin{figure}[h!]
  \centering
  \includegraphics[scale=0.9]{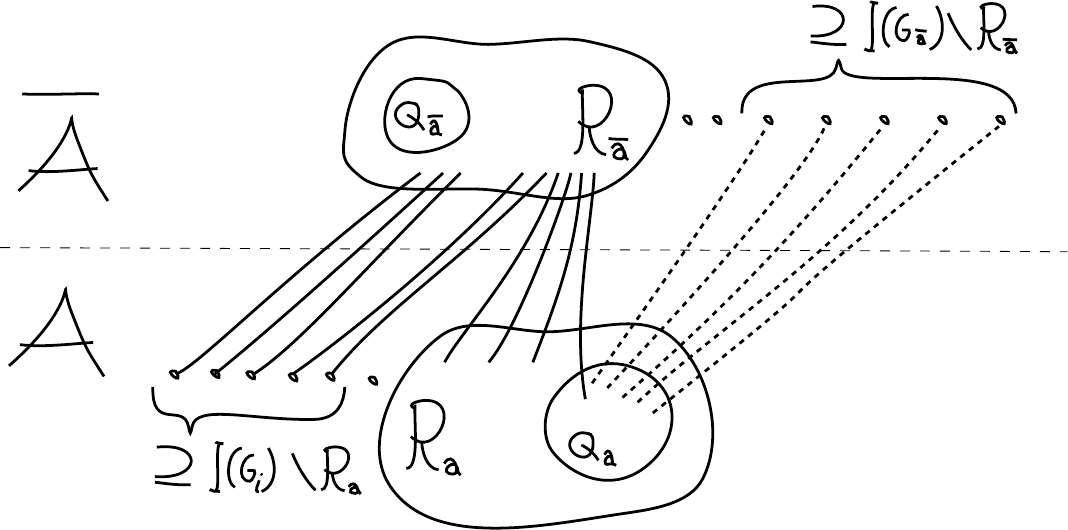}
  \caption{As described in the proof of
    Lemma~\ref{lemma:eds:trim:non-split}. The dotted lines constitute
    the set $E_I$.}
  \label{fig:eds:trim}
\end{figure}

\begin{lemma}    \label{lemma:eds:trim:non-split}    Procedure    {\sc
    Trim$_{\textsc{eds}}$-non-split} produces  a set $S'  \better_A S$
  of size at most $3^{\mm(A)}$ in time $\bigoh^*(3^{\mm(A)}|S|)$.
\end{lemma}

\begin{qedproof}
  The size of  $S'$ is apparent from the fact  that we iterate through
  all sets $Q \subseteq R \subseteq C$ (at most $3^{\mm(A)}$ triples),
  and  mark one  certificate in  $S$  to later  be put  in the  output
  set.  When deciding which  certificate to  mark, we  possibly search
  through the  entire set $S$, and  do a polynomial amount  of work on
  each  certificate in  $S$ to  check if  it is  the  certificate that
  should be  marked in this iteration.  This gives a  total runtime of
  $\bigoh^*(3^{\mm(A)}|S|)$.
  
  By the definition of the $\better_A$-relation, we have $S' \better_A
  S$ if for every important  certificate $G_a \in S$ and $G_{\compl a}
  \in \cert(\compl  A)$ so that $\conc(G_a, G_{\compl  a})$ contains a
  witness, the set  $\conc(S', \set{ G_{\compl a} })$  also contains a
  witness.

  Let $G_w = G_a + G_{\compl a} + E_w$ be a witness of $\conc(G_a,
  G_{\compl a})$.  %
  We show that the algorithm will mark a certificate $G_i$ so that the
  set $\conc(G_i, G_{\compl a})$ also contains a witness, thereby
  establishing the lemma.
  By our definition of  $\conc$, the set  $E_w$ is a
  subset of edges  from $E( G[V(G_a), V(G_{\compl a})] )$. Let  $R'$ be the
  vertices of  $V(G_w) \cap  C$ and  let $E_I$ be  the edges  of $E_w$
  where one  endpoint is in  $C \cap A$  and the other endpoint  is in
  $\compl  A  \setminus  C$.  Note   that  $E_I$  must  span  the  set
  $I(G_{\compl a}) \setminus R'$. We  let $Q_a$ be the set of vertices
  in  $R' \cap  A$ that  are  incident with  edges of  $E_I$, and  let
  $Q_{\compl a}$  be the vertices of  $R' \cap \compl  A$ not incident
  with any edge in $E_w$ (see Figure~\ref{fig:eds:trim}).
  
  At the iteration of the algorithm  where $Q = Q_a \cup Q_{\compl a}$
  and $R = R'$, the algorithm marks a certificate $G_i$ to be put into
  $S'$.  This means  for the  certificate $G_i$  there is  a  set $E_i
  \subseteq E( G[V(G_i),  R_{\compl a}] )$ spanning  $I(G_i) \setminus Q_a$
  and  $R_{\compl  a}\setminus  Q_{\compl  a}$. Furthermore,  we  have
  $|E_i| + |E(G_i)| \leq|E_w  \setminus E_I| + |E(G_a)|$, as otherwise
  the algorithm would prefer marking $G_a$ over marking $G_i$.

  We will now show that the certificate $G_w' = G_i + G_{\compl a} +
  E_I + E_i \in \conc(G_i, G_{\compl a})$ is a witness. To do this, it
  will be enough to show the following three things: (1) $V(G_w')$ is
  a vertex cover of the original graph $G$, (2) $|E(G_w')| \leq |E(G_w)|$,
  and (3) $I(G_w')$ is empty.

  The first point, that $V(G_w')$ is a vertex cover, we get from the
  fact that $G_w'$ is a vertex cover of $G[A, \compl A]$ and that as
  both $G_i$ and $G_{\compl a}$ are locally correct $V(G_i)$ and
  $V(G_{\compl a})$ are vertex covers of $G[A]$ and $G[\compl A]$,
  respectively. So, $G_w'$ is a vertex cover of $G[A] + G[\compl A] +
  G[A, \compl A] = G$.

  The second point, $|E(G_w')| \leq |E(G_w)|$, holds by the following
  inequality.
  \begin{align*}
    |E(G_w')| &= |E(G_i)|  + |E_i| + |E(G_{\compl a})| + |E_I|\\
    &\leq |E(G_a)| + |E_w \setminus E_I| + |E(G_{\compl a})| + |E_I| =
    |E(G_w)|\enspace .
  \end{align*}

  What remains to prove is the third point, that $I(G_w')$ is
  empty. To do this, we will show that each vertex in $I(G_i)$ and
  $I(G_{\compl a})$ is spanned by the edges $E_I + E_i$.
  We defined the vertices of $Q_{\compl a}$ to not be incident with
  any edges of $E_w$, and we know $I(G_w)$ is empty, so we can
  conclude that $Q \cap I(G_{\compl a})$ is empty also.  This means
  that as $E_i$ spans the set $R \cap \compl A \setminus Q$, in fact
  $E_i$ spans all the vertices of $I(G_{\compl a}) \cap R$. We defined
  $E_I$ to be all the edges of $E_w$ with endpoints in $\compl A
  \setminus R$, and $E_w$ spans $I(G_{\compl a})$, so $E_I$ must span
  $I(G_{\compl a}) \setminus R$.  From this we conclude that $I(G_w')
  \setminus A$ is empty.
  Now we must show that also  $I(G_w') \cap A$ is empty. We know $E_i$
  spans  all the  vertices  of  $I(G_w') \cap  A  \setminus Q_{a}$  by
  definition, and the set $Q_a$ is  exactly the set of vertices in $A$
  that are incident with edges of $E_I$. So, $E_I \cup E_i$ also spans
  $I(G_w') \cap A$, and hence $I(G_w') = \emptyset$.
\end{qedproof}

Next, we give the following lemma, which will be used directly in the
join operation of our algorithm solving $t$-{\sc Edge Dominating
  Set}. It will be helpful in reducing the number of certificates
needed to ensure we have a set preserving $\conc(G_1, G_2)$ for given
certificates $G_1$ and $G_2$. 
The idea behind this lemma is that for any witness $G_1 + G_2 + G_3 +
E_w$, we can substitute the set $E_w$ with another set $E_w'$ spanning
all the isolated vertices of $G_1, G_2$ and $G_3$ to produce a new
witness, as long as $|E_w'| \leq |E_w|$. This means that a lot of
different certificates $G_1 + G_2 + E_t \in \conc(G_1, G_2)$ will
generate a witness when combined with the same certificate $G_3$ as
long as there is a set $E_t'$ so that $E_t' + E_t$ spans the isolated
vertices of $G_1, G_2$ and $G_3$. Having this in mind, we are able to
limit the number of certificates needed to preserve $\conc(G_1, G_2)$
greatly.

\begin{lemma} \label{lemma:eds:rep}
  For any disjoint sets $A_1, A_2 \subseteq V(G)$ and certificates
  $G_1 \in \cert(A_1)$ and $G_2 \in \cert(A_2)$, we can in
  $\bigoh^*(2^{\sm(A_1)} + 2^{\sm(A_2)})$-time compute two set
  families $\edsc(G_1, G_2) \subseteq 2^{V(G_1)}$ and $\edsc(G_2,G_1)
  \subseteq 2^{V(G_2)}$ where the following holds:
  \begin{enumerate}
  \item $|\edsc(G_1,G_2)| \leq \max{\set{2^{\sm(A_1)},n}}$ and $|\edsc(G_2,G_1)| \leq \max{\set{2^{\sm(A_2)},n}}$, and
  \item there is a set $S \subseteq \conc(G_1, G_2)$ preserving
    $\conc(G_1, G_2)$, where for each certificate $G_1 + G_2 + E_s \in
    S$ the set $E_s$ has $V(E_s) \cap A_1 \in \edsc(G_1 , G_2)$ and
    $V(E_s) \cap A_2 \in \edsc(G_2 , G_1)$.
  \end{enumerate}
\end{lemma}

\begin{qedproof}
  Let $A_3 = \compl {A_1 \cup A_2}$.  We will give a construction of
  the set families $\edsc(G_1 , G_2)$ and show that for any
  certificate $G_z = G_1 + G_2 + E_z$ in $\conc(G_1, G_2)$, there is a
  certificate $G_z' = G_1 + G_2 + E_z'$ in $\conc(G_1, G_2)$ so that
  $G_z' \better_{\compl {A_3}} G_z$ and that $V(E_z') \cap A_2 =
  V(E_z) \cap A_2$ and $V(E_z') \cap A_1 \in \edsc(G_1 , G_2)$. By
  similar construction and argument for $\edsc(G_2 , G_2)$, we can
  conclude that constraint (\emph{2.})  holds for the two constructed
  set families. That the two set families can be constructed within
  the proposed time bound and that the size of the sets are as stated
  in constraint (\emph{1.}) is evident from how we are going to
  construct them.

  Suppose for certificate $G_z = G_1 + G_2 + E_z \in \conc(G_1, G_2)$
  there is some certificate $G_3 \in \cert(A_3)$ so
  that we have a witness $G_w = G_1 + G_2 + G_3 + E_w$ in $\conc(G_3,
  G_z)$. Let $C$ be a vertex cover of $G[{A_1}, \compl {A_1}]$. We are
  particularly interested in the following three parts of $E_w$
  incident with $A_1$ (see Figure~\ref{fig:eds:setFamily}).
  \begin{itemize}
  \item $E_1$: the subset of edges that go from $A_1 \cap C$ to $A_2$,
  \item $E_2$: the subset of edges between $A_1 \setminus C$ and $A_2
    \cap C$, and
  \item $E_3$: the edges that go from $A_1 \setminus C$ to $A_3 \cap
    C$.
  \end{itemize}
  We denote by $R_1$, $R_2$ and $R_3$ the endpoints $A_1 \cap V(E_1)$,
  $A_2 \cap V(E_2)$, and $A_3 \cap V(E_3)$, respectively.

  As $G_w$ is a witness, the edges in $E_w$ must span all the isolated
  vertices of $G_1$, $G_2$, and $G_3$. We notice that the edges
  $E(G_w) \setminus (E_2 \cup E_3)$ span all the vertices of $V(G_w)$
  except possibly some vertices in $I(G_1) \cup R_2 \cup
  R_3$. Therefore, for any subset $E' \subseteq E(I(G_1)\setminus R_1,
  \;R_2 \cup R_3)$ spanning $(I(G_1)\setminus R_1) \cup R_2 \cup R_3$,
  it will be the case that $E(G_w) + E' - E_2 - E_3$ spans
  $V(G_w)$. Furthermore, if $\card{E'} \leq \card{E_2 \cup E_3}$, then
  the certificate $G_z' = G_w - E_2 - E_3 + E'$ will be a witness.  To
  ensure that constraint (\emph{2.}) is satisfied, it thus suffice to
  ensure that the vertices in $A_1$ adjacent to vertices in $A_2$ in
  $G_z'$ constitute a set in $\edsc(G_1 , G_2)$. For $G_z'$, these
  vertices are exactly $R_1 \cup (V(E') \cap A_1)$. What we notice is
  that $E'$ only depended on $G_1, R_1, R_2$ and $R_3$. So, if for
  each choice of $R_1$, $R_2$, and $R_3$, we compute the $E'$ as we
  did above, we will have made a set family satisfying constraint
  (\emph{2.}) and which can be computed within a polynomial factor in
  $n$ of its size.

  What we may now observe, is that since $R_1 \subseteq C \cap A_1$,
  $R_2 \subseteq C \cap A_2$ and $R_3 \subseteq C \cap A_3$, the
  different possibilities for $R_1, R_2$ and $R_3$ combined is at most
  $2^{|C|}$. This means the size of the set family is at most
  $2^{\mm(A_1)}$ and the time to compute it is
  $\bigoh^*(2^{\mm(A_1)})$.

  If $\scut{A_1}$ is a split, however, it might be the case that
  $\sm(A_1) < \mm(A_1)$. But when it is a split, as the neighbourhood
  of each vertex in $I(G_1)$ is the same, as long as there is a set
  $S_i \in \edsc(G_1,{G_2})$ of $i$ vertices maximizing
  $\max\{\card{S_i \cap I(G_1)}\}$ for each $i \in \{0, \ldots,
  \card{V(G_1)}\}$, the set family $\edsc(G_1,{G_2})$ satisfies constraint
  (\emph{2.}). The size constraint follows from the fact that we only
  need at most $n$ sets $S_i$, and clearly we can compute the set
  family in the runtime stated, as it only takes polynomial amount of
  time to generate $S_i$ greedily for each $i$.
\end{qedproof}

\begin{figure}[h!]
  \centering
  \includegraphics{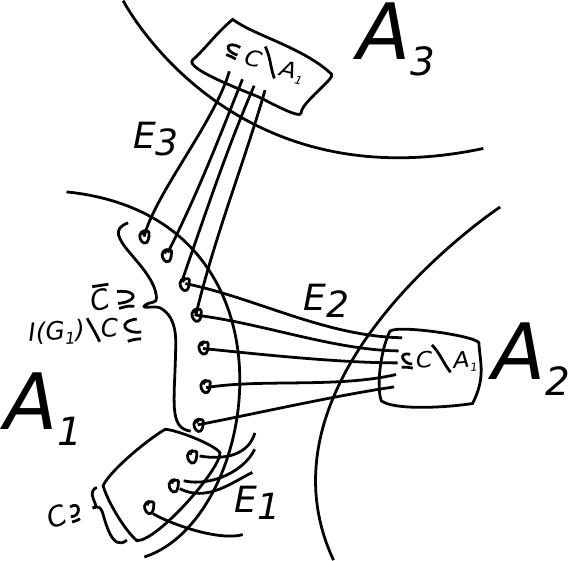}
  \caption{As described in proof of Lemma~\ref{lemma:eds:rep}. The
    vertices in each of the three rectangles are subsets of the vertex
    cover $C$ of $A_1$, so these sets can be of at most $2^{|C|}$
    possibilities.}
  \label{fig:eds:setFamily}
\end{figure}

\begin{center}
\begin{minipage}{1.0\textwidth}
\begin{tabbing} 
  \rule{\textwidth}{0.25pt}\\ 
  {\bf Procedure {\sc Join$_{\textsc{eds}}$}} 
  (\= on node $w$ with children $a,b$ and $A_1=V_a, A_2 =V_b$\\
   \> and $A=A_1 \cup A_2$ and given $S_1 \better_{A_1} \cert(A_1)$ and\\
   \> $S_2 \better_{A_2} \cert(A_2)$)\\
  \rule{\textwidth}{0.25pt}\\ 
  x\tabme{}
  \> $S \leftarrow \emptyset$ \\
  \> \textbf{for} each $G_1 \in S_1$, $G_2 \in S_2$ \textbf{do}\\
  \> \> $\mathcal{V}_1 \leftarrow \edsc(G_1 , G_2)$ from Lemma \ref{lemma:eds:rep}\\
  \> \> $\mathcal{V}_2 \leftarrow \edsc(G_2 , G_1)$ from Lemma \ref{lemma:eds:rep}\\
  \> \> \textbf{for} each $X_1 \in \mathcal{V}_1$, $X_2 \in \mathcal{V}_2$ 
        \textbf{do}\\
  \> \> \> $E' \leftarrow$ minimum sized subset of $E( G[X_1, X_2] )$ 
           spanning $X_1 \cup X_2$\\
  \> \> \> add to $S$ the certificate $G_1 + G_2 + E'$\\[6pt]
  
  \> \textbf{if} $\scut {A}$ is a split \textbf{then}
  \textbf{return} \textsc{Trim$_{\textsc{eds}}$-split}($S \subseteq \cert(A)$)\\
  \> \textbf{else} \textbf{return} 
     \textsc{Trim$_{\textsc{eds}}$-non-split}($S \subseteq \cert(A)$)\\
  \rule{\textwidth}{0.25pt}
\end{tabbing}
\end{minipage}
\end{center}

\begin{lemma}\label{lemma:eds:join}
  The algorithm \textsc{Join$_{\textsc{eds}}$} is correct and runs in
  time $\bigoh^*(\card{S_1}\card{S_2}12^{k})$, producing a set $S$ of
  cardinality $\bigoh(n^2 + 3^k)$ when both $\sm(A_1)$, $\sm(A_2)$,
  and ${\sm(A_1 \cup A_2)}$ is at most $k$.
\end{lemma}

\begin{qedproof}
  For $G_1 \in \cert(A_1)$ and $G_2 \in \cert(A_2)$, if $G^* = G_1 +
  G_2 + E^* \in \conc(G_1, G_2)$ and $G' = G_1 + G_2 + E'$ where
  $V(E') = V(E^*)$, then $G' \better_{A_1 \cup A_2} G^*$ as long as
  $\card{V(E')} \leq \card{V(E^*)}$. Therefore, the set of
  certificates $S$ generated in the first part of the algorithm must
  preserve the set ``S'' from point 2.\ in the statement of
  Lemma~\ref{lemma:eds:rep}, which in turn implies that $S
  \better_{A_1 \cup A_2} \conc(S_1, S_2)$. By
  Lemma~\ref{lemma:eds:trim:split} and
  Lemma~\ref{lemma:eds:trim:non-split} this means the output of the
  algorithm is a set of size at most $n^2 + 3^{\sm(A_1 \cup A_2)}$
  that preserves $\conc(S_1, S_2)$.

  From Lemma~\ref{lemma:eds:rep}, for each pair of certificates, it
  takes $\bigoh^*(2^{k})$ time to compute $\edsc(G_1,G_2)$ and $\edsc(G_2,G_1)$, and for
  each pair we generate at most
  $\card{\edsc(G_1,G_2)}\card{\edsc(G_2,G_1)}$ certificates. So, before
  the call to one of the \textsc{Trim$_{\textsc{eds}}$}-procedures,
  the algorithm uses $\bigoh^*((2^{2k})\card{S_1}\card{S_2})$-time and
  $S$ contains at most $(2^{2k})\card{S_1}\card{S_2}$ certificates.
  The total runtime, including the call to the respective
  \textsc{Trim$_{\textsc{eds}}$}-procedure, must then by
  Lemma~~\ref{lemma:eds:trim:split} and
  Lemma~\ref{lemma:eds:trim:non-split} be
  $\bigoh^*(3^{k}2^{2k}\card{S_1}\card{S_2})$.
\end{qedproof}

\begin{theorem} \label{thm:eds} Given a graph $G$ and branch
  decomposition $(T, \delta)$ of \SMW\ $k$, we can solve %
  \textsc{Edge Dominating Set} in time $\bigoh^*(3^{5k})$.
\end{theorem}
\begin{qedproof}
  In Lemma~\ref{lemma:eds:join} we showed that the procedure
  \textsc{Join$_{\textsc{eds}}$} is correct and runs in time
  $\bigoh^*(\card{S_1}\card{S_2}12^{k})$, producing a set $S$ of
  cardinality $\bigoh(n^2 + 3^k)$. So, using {\sc Recursive} with {\sc
    Join$_{\textsc{eds}}$}, we know the size of both of the inputs of
  {\sc Join$_{\textsc{eds}}$} is at most the size of its output, i.e.,
  $|S_1|, |S_2| \leq \bigoh^*(3^{5k})$. So, each call to {\sc
    Recursive} has runtime at most $\bigoh^*(3^{k})$. As there are
  linearly many calls to {\sc Recursive} and there is a polynomial
  time verifier for the certificates {\sc Recursive} produces, by the
  definition of $\better$, the total runtime is also bounded by
  $\bigoh^*(3^{5k})$. To solve {\sc Edge Dominating Set}, we run the
  $t$-{\sc Edge Dominating Set} algorithm for all values of $t \leq n$
  and hence this is also solvable in $\bigoh^*(3^{5k})$ time as we
  exclude polynomials of $n$.
\end{qedproof}
\section{Graphs of bounded sm-width}

In this section we discuss graph classes of bounded sm-width.

\newcounter{props}\setcounter{props}{\value{theorem}}
\begin{proposition}\label{prop:1}
  If treewidth is bounded then sm-width is bounded $(\smw(G) \leq
  tw(G)+1)$ which in turn means that clique-width is bounded.
\end{proposition}
\begin{proposition}\label{prop:2}
  If twin-cover is bounded then clique-width and sm-width is bounded
  ($\smw(G) \leq tc(G)$).
\end{proposition}
\begin{proposition}\label{prop:3}
  Cographs, the graphs of clique-width at most two, have sm-width one,
  while distance-hereditary graphs have clique-width at most three and
  sm-width one.
\end{proposition}

\begin{proof}[Proof of Proposition~\ref{prop:1},\ref{prop:2} and \ref{prop:3}]
Firstly, by  the definition of sm-width
and MM-width it is clear that  for any graph $G$ we have $\smw(G) \leq
\mmw(G)$ and it then  follows by Theorem \ref{vatshelle} (by Vatshelle
\cite{Vatshelle}) that $\smw(G) \leq  \tw(G)+1$.  Let us argue that if
sm-width is  bounded on  a class of  graphs, then so  is clique-width.
Rao \cite{Rao} shows that the clique-width of a graph is at most twice
the maximum clique-width of all  prime graphs in a split decomposition
of the  graph. By  Observation \ref{obs:inducedSubgraph} and  the fact
that every prime graph is an induced subgraph we know that sm-width of
a graph  is at  least as large  as the  maximum sm-width of  all prime
graphs.  Theorem \ref{vatshelle} and the fact that sm-width of a prime
graph is equal  to the MM-width of the prime graph,  tells us that the
sm-width  of  a prime  graph  is  bounded  whenever its  treewidth  is
bounded.  Thus, since clique-width  is stronger  than treewidth  it is
also stronger than sm-width.

Secondly, twin-cover $tc(G)$ is a graph parameter introduced by Ganian
\cite{GanianIPEC2011}  as a  generalization  of vertex  cover that  is
bounded also  for some  dense classes of  graphs.  Gajarsk{\'y}  et al
\cite{IPEC2013} in their study  of the modular-width parameter $mw(G)$
showed that $mw(G) \leq 2^{tc(G)}+tc(G)$ which in turn implies that if
twin-cover is bounded for a  class of graphs then also clique-width is
bounded.  We show that $\smw(G)  \leq tc(G)$.  Using the definition of
$tc(G)$ from  \cite{GanianIPEC2011} it follows  that $G$ has a  set $S
\subseteq V(G)$ of at most  $tc(G)$ vertices such that every component
$C$ of  $G \setminus S$ induces a  clique and every vertex  in $C$ has
the same neighborhood in $S$.   Let $C_1,...,C_q$ be the components of
$G \setminus S$.  Take any branch decomposition of $G$ having for each
component  $C_i$ a subtree  $T_i$ such  that the  leaves of  $T_i$ are
mapped to the vertices of $C_i$  and also having a subtree $T_S$ whose
leaves are mapped to  $S$. The cuts of $G$ induced by  an edge of this
branch decomposition  are of three  types depending on where  the edge
is: if it is inside the tree  of $S$ the cut has a maximum matching of
size at most $|S|\leq tc(G)$; if it  is inside a tree $T_i$ the cut is
a split; otherwise the cut has a maximum matching of size at most $|S|
\leq tc(G)$.

Thirdly,  any cograph  is  also a  distance-hereditary  graph and  any
distance-hereditary  graph $G$  has  clique-width at  most three,  see
e.g.  \cite{HOSG}. Also, $G$  has a  branch decomposition  of cut-rank
one, as shown by Oum \cite{Oum},  which means that all cuts induced by
edges  of   this  branch  decomposition  are  splits   and  hence  any
distance-hereditary graph has sm-width one. 
\end{proof}

There are several classes of graphs of bounded sm-width where no previous results
implied FPT algorithms for the considered problems. 
We now show a class of such graphs, constructed by combining a graph of 
clique-width at most 3, with a graph of treewidth $k$ and thus
clique-width at most $2^{k/2}$, as follows.  Let $G_1$ be a
distance-hereditary graph and let $G_2$ be a graph of treewidth $k$.
Let $X \subseteq V(G_1)$ with $|X| \leq k+1$ and $(X, \overline{X})$ a
split of $G_1$, and let $Y \subseteq V(G_2)$ be a bag of a tree
decomposition of $G_2$ of treewidth $k$.
Add  an arbitrary set  of edges  on the  vertex set  $X \cup  Y$.  The
resulting  graph will  have  sm-width  at most  $k+1$,  a result  that
basically follows  by taking branch decompositions of  $G_1$ and $G_2$
where $X$ and $Y$  each are mapped as the set of  leaves of a subtree,
subdividing each of  the two edges above these  subtrees and adding an
edge on the subdivided vertices  to make a single branch decomposition
of the combined graph.

Note that we can also construct new tractable classes of graphs by
combining several graphs in a tree structure.

\section{Conclusions} \label{sec:conclusion}

We have shown that four basic problems, that cannot be FPT
parameterized by clique-width unless FPT~=~W[1], are FPT when
parameterized by the split-matching-width of the graph, a parameter
whose modelling power is weaker than clique-width but stronger than
treewidth.  This was accomplished using the theory of split
decompositions and the recently introduced MM-width,
combined with slightly non-standard dynamic programming algorithms on
the resulting decompositions. Graph classes of bounded sm-width will
consist of graphs having low treewidth in local parts, with these
parts connected together in a dense manner.  We have not found
references to such graph classes in the litterature.  Nevertheless,
the appealing algorithmic properties of these graph classes should
merit an interest in their further study.

\bibliography{sm}
\end{document}